\documentclass[a4paper]{article}

\usepackage[dvipdfm]{graphicx}
\usepackage{amsthm}
\usepackage{amsmath}
\usepackage{amssymb}
\usepackage{color, array}
\usepackage{url}
\usepackage{multirow}
\usepackage{algorithmic}
\usepackage{algorithm}
\usepackage{dcolumn}
\usepackage{tikz}

\newtheorem{Definition}{Definition}[section]
\newtheorem{Theorem}[Definition]{Theorem}
\newtheorem{Lemma}[Definition]{Lemma}
\newtheorem{Corollary}[Definition]{Corollary}
\newtheorem{Example}[Definition]{Example}
\newtheorem{Remark}[Definition]{Remark}
\newtheorem{Proposition}[Definition]{Proposition}

\newcommand{\lt}{{\mathop{\mathrm{lt}}}}
\newcommand{\lc}{{\mathop{\mathrm{lc}}}}
\newcommand{\lm}{{\mathop{\mathrm{lm}}}}

\begin{document}
\title{Comprehensive Systems for Primary Decompositions of Parametric Ideals}

\author{Yuki Ishihara\thanks{Nihon University, ishihara.yuki@nihon-u.ac.jp} \and Kazurhiro Yokoyama\thanks{Rikkyo University, kazuhiro@rikkyo.ac.jp}}
\date{}

\maketitle             

\begin{abstract}
We present an effective method for computing parametric primary decomposition via comprehensive Gr\"obner systems. In general, it is very difficult to compute a parametric primary decomposition of a given ideal in the polynomial ring with rational coefficients $\mathbb{Q}[A,X]$ where $A$ is the set of parameters and $X$ is the set of ordinary variables. One cause of the difficulty is related to the irreducibility of the specialized polynomial. Thus, we introduce a new notion of ``feasibility'' on the stability of the structure of the ideal in terms of its primary decomposition, and we give a new algorithm for computing a so-called {\em comprehensive system} consisting of pairs $(C, \mathcal{Q})$, where for each parameter value in $C$, the ideal has the stable decomposition $\mathcal{Q}$. We may call this comprehensive system {\em a parametric primary decomposition} of the ideal. Also, one can also compute a dense set $\mathcal{O}$ such that $\varphi_\alpha(\mathcal{Q})$ is a primary decomposition for any $\alpha\in C\cap \mathcal{O}$ via irreducible polynomials. In addition, we give several computational examples to examine the effectiveness of our new decomposition.
\end{abstract}

\section{Introduction}

In the analysis of the structure of a polynomial ideal, primary decomposition plays an important role. For example, the primary decomposition of a radical ideal corresponds to the irreducible decomposition of its variety. In engineering and pure mathematics computations, it is often required to deal with parametric ideals, i.e., ideals generated by polynomials with parametric coefficients. Thus, for such parametric ideals, algorithms for their Gr\"obner bases and several ideal operations have been studied in \cite{Kapur,Montes,Nabeshima-Tajima,Yokoyama2006}. In particular, we sometimes face on a family of ideals described by parameters and it is highly desirable to classify their primary decompositions with respect to parameter values to analyze their algebraic structure. However, in general, it is very difficult to comprehensively classify parameter values for which the given ideal has its primary decomposition uniformly. In this paper, we introduce a new notion ``feasible comprehensive primary decomposition system'' based on the Hilbert's irreducible theorem that can resolve the difficulty partly. Then, we propose a new approach to compute such classification (feasible comprehensive primary decomposition system) and give an effective method via Comprehensive Gr\"obner Systems (CGSs in short). 

Algorithms for ordinary primary decomposition (without parameters) have been studied in \cite{GIANNI1988149,Eisenbud1992,KAWAZOE20111158,SHIMOYAMA1996247}. Throughout this paper, we let $I$ be an ideal of $K[A,X]$, where $K$ is the rational number field $\mathbb{Q}$ or algebraic extensions of $\mathbb{Q}$. We also let $K[A,X]$ be the polynomial ring over $K$, where $A=\{a_1,\ldots,a_m\}$ is the set of parameters and $X=\{x_1,\ldots,x_n\}$ is the set of ordinary variables. For $\alpha\in K^m$ we denote by $\varphi_\alpha$ the canonical specialization homomorphism from $K[A,X]$ to $K[X]$, i.e.,
$\varphi_\alpha (f(A,X))=f(\alpha,X)$ for $f(A,X)\in K[A,X]$. We say that, with respect to a specified property, a parametric ideal $I$ is {\em stable} on a subset $C$ of $K^m$ if $\varphi_\alpha(I)$ keeps the property for any $\alpha\in C$. We consider primary decomposition as such a specified property and set our main goal to develop an efficient method to compute {\em a parametric primary decomposition}. Our contributions are as follows: 

\begin{enumerate}
	\item Define a {\em feasible} comprehensive primary decomposition system (feasible CPDS in short) of a parametric ideal (see Definition \ref{def:fCPDS}).
	\item Devise effective algorithms to compute feasible CPDS and minimal feasible CPDS (see Algorithms~\ref{alg1} and \ref{alg2}).
    \item Prove the density of feasible CPDS by using Hilbert's Irreducibility Theorem (see Section~\ref{sec5}).
    \item Devise a new method that uses minimal polynomials to derive the conditions for a specialized ideal to become an prime ideal (see Algorithm~3)
	\item Implement the algorithm in the computer algebra system Risa/Asir and examine its performance in several examples (see Section 6).
\end{enumerate}

We briefly outline (ordinary) primary decomposition. A set of ideals $\{Q_1,$ $\ldots$ $,Q_r\}$ of $K[A,X]$ is called a {\em primary decomposition} of $I$ if it satisfies the following conditions: 
\begin{enumerate}
\setlength{\leftskip}{1cm}
	\item[(PD-1)]$Q_1,\ldots Q_{r-1}$ and $Q_r$ are primary ideals of $K[A,X]$, 
	\item[(PD-2)] $I=Q_1\cap \cdots \cap Q_r$. 
\end{enumerate}
A primary decomposition of $I$ is said to be {\em minimal} or {\em irredundant} if it satisfies the following additional conditions: 
\begin{enumerate}\setlength{\leftskip}{1cm} \setcounter{enumi}{2}
		\item[(M-1)] $\bigcap_{j\neq i} Q_j\not \subset Q_i$ for all $i$,
		\item[(M-2)] $\sqrt{Q_i}\neq \sqrt{Q_j}$ for all $i\neq j$.
\end{enumerate}
We note that for a minimal primary decomposition $\mathcal{Q}=\{Q_1,\ldots,Q_r\}$ of $I$, $\varphi_\alpha(\mathcal{Q})=\{\varphi_{\alpha} (Q_1),$ $\ldots,\varphi_{\alpha} (Q_r)\}$ is not always that of $\varphi_\alpha (I)$. For instance, $I=\langle (x_1+a_1)x_1^2,(x_1+a_1)x_2 \rangle$, the ideal generated by $(x_1+a_1)x_1^2$ and $(x_1+a_1)x_2$ in $\mathbb{Q}[a_1,x_1,x_2]$, has a minimal primary decomposition $\mathcal{Q}=\{\langle x_1+a_1\rangle,\langle x_1^2,x_2\rangle\}$ but $\varphi_0(\mathcal{Q})=\{\langle x_1 \rangle,\langle x_1^2,x_2 \rangle\}$ is not a minimal primary decomposition of $\varphi_0(I)=\langle x_1^3,x_1x_2\rangle=\langle x_1^3 \rangle\cap \langle x_1,x_2 \rangle$. To solve this problem, we split $K^m$ into subsets $C_1,\ldots,C_l$ such that $\bigcup_{i=1}^{l} C_i=K^m$ and the primary decomposition of $\varphi_\alpha (I)$ is stable for $\alpha\in C_i$ in each $i=1,\ldots,l$. We might call idealistically $\{(C_1,\mathcal{Q}_1),\ldots,(C_l,\mathcal{Q}_l)\}$ a {\em comprehensive primary decomposition system} (CPDS) for $I$ if $\varphi_\alpha(\mathcal{Q}_i)$ is a primary decomposition of $\varphi_\alpha(I)$ for any $\alpha\in C_i$ and each $i$. For dealing with CPDS algebraically each subset $C_i$ is desirable be a finite union of {\em locally closed set} (see Definition \ref{def:ACS}), i.e., a form  $V(J_1)\setminus V(J_2)$ where $V(J_1)$ and $V(J_2)$ are varieties of ideals $J_1$ and $J_2$ in $K[A]$ respectively. 

In order to compute a CPDS, we utilize a comprehensive Gr\"obner system (see Definition \ref{def:CGS}). A comprehensive decomposition system satisfying conditions (PD-2), (M-1) and (M-2) can be computed completely, however, one satisfying (PD-1) does not necessarily exist. For example, for $I=\langle x_1^2-a_1 \rangle \subset \mathbb{Q}[a_1,x_1]$, $\varphi_\alpha (I)=\langle x_1^2-\alpha \rangle$ is a prime ideal if and only if $x^2-\alpha$ is irreducible over $\mathbb{Q}$, that is, $\alpha$ is not a positive square number; thus the condition can not determine a locally closed set. Hence, we introduce a notion of ``feasible'' CPDS based on Hilbert's irreducibility theorem, that is, for an irreducible polynomial $f(A,X)$ in $\mathbb{Q}[A,X]$, there exists a dense set $\mathcal{O}$ in $\mathbb{Q}^m$ such that $\varphi_\alpha (f)=f(\alpha,X)$ is irreducible over $\mathbb{Q}$ for $\alpha\in \mathcal{O}$ . We generalize this theorem for prime ideals and primary ideals (see Theorem \ref{thm:hilb_prime} and Corollary \ref{cor:hilb_primary}). 

For a practical computation of a feasible CPDS of $I$, as the first step, we compute a primary decomposition $\mathcal{Q}=\{Q_1,\ldots,Q_r\}$ of $I$ in $\mathbb{Q}[A,X]$. As the second step, we extract algebraic conditions on $A$ so that a given set of ideals $\mathcal{Q}=\{Q_1,\ldots,Q_r\}$ satisfies conditions (PD-2), based on CGS-computations. By recursively executing these two steps, we obtain a feasible CPDS (see Algorithm \ref{alg1}). For a minimal feasible CPDS, we need additional steps for satisfying conditions (M-1) and (M-2) (see Algorithm \ref{alg2}). For a prime ideal $P$ with parameters, we utilize minimal polynomials of (a localized) $P$ to compute the conditions of $\alpha\in K^m$ for $\varphi_\alpha(P)$ to be a prime ideal. We also implement our algorithms in the computer algebra system Risa/Asir and examine it in several examples. 

This paper is organized as follows. In section 2, we recall some notions about primary decomposition and comprehensive Gr\"obner system. In section 3, we define the stability of parametric ideal operations and feasible CPDS for a parametric ideal. In section 4, we introduce the main algorithms for feasible CPDSs and prove each correctness and termination. In section 5, we discuss the primarity of parametric primary ideal based on the Hilbert's irreducibility theorem. In section 6, we see some computational examples and the effectiveness of our algorithm. In section 7, we explain the conclusion and future works. 

\section{Mathematical Basis}
Let $K$ be a computable field of characteristic $0$ such as $\mathbb{Q}$ and algebraic extensions of $\mathbb{Q}$. For a set of parameters $A=\{a_1,\ldots,a_m\}$ and a set of ordinal variables $X=\{x_1,\ldots,x_n\}$, we denote by $K[A,X]$ and $K[X]$ the parametric polynomial ring and the ordinary polynomial ring over $K$, respectively. For $\alpha\in K^m$, we define the canonical specialization homomorphism $\varphi_\alpha$ from $K[A,X]$ to $K[X]$ by $\varphi_\alpha (f(A,X))=f(\alpha,X)$ and let $\varphi_\alpha (\mathcal{Q})=\{\varphi_\alpha (Q)\mid Q\in\mathcal{Q}\}$ for a set of ideals $\mathcal{Q}$. We note that, for an ideal $Q$,  $\varphi_\alpha(Q)=\{f(\alpha,X)\in K[X]\mid f(A,X)\in Q\}$ forms an ideal of $K[X]$. For the ring $R=K[A,X]$ or $K[X]$, we denote by $\langle f_1,\ldots,f_s \rangle_R$ the ideal generated by $f_1,\ldots,f_s$ in $R$. If $R$ is obvious, then we simply write $\langle f_1,\ldots,f_s \rangle$. Throughout this paper, we let $I$ be a proper ideal of $R$ and $\sqrt{I}$ the radical of $I$, i.e., $\sqrt{I}=\{f\in R\mid f^m\in I \text{ for some $m\in \mathbb{N}$}\}$ where $\mathbb{N}$ is the set of positive integers. Also, for an ideal $J$ of $K[A]$, we denote by $V_K(J)$ the variety of $J$ over $K$, i.e., $V_K(J)=\{\alpha\in K^m\mid f(\alpha)=0 \text{ for all $f\in J$}\}$ and simply write $V(J)$ if $K$ is fixed. Let $T(A)$, $T(X)$ and  $T(A,X)$ be the set of monomials of $K[A]$, $K[X]$ and $K[A,X]$ respectively. When we deal with Gr\"obner bases, we assume that a monomial ordering is fixed a priori. For a monomial ordering $\succ$ on $T(A,X)$, $\succ$ is called a block ordering with  $X\succ\succ A$ if any monomial in $T(A)$ is less than any monomial in $T(X)$. For a block ordering $\succ$ with $X\succ\succ A$, $\succ_X$ and $\succ_A$ denote restrictions of $\succ$ on $T(X)$ and $T(A)$ respectively. For a non-zero polynomial $f$ of $K[A,X]$, we denote by $\lc_{\succ}(f)$, $\lm_{\succ}(f)$ and $\lt_{\succ}(f)$ the leading coefficient of $f$, the leading monomial of $f$ and the leading term of $f$ with respect to $\succ_X$ in $K[A][X]$ respectively, i.e., $\lc_\succ(f)\in K[A]$, $\lm_\succ(f)\in K[X]$ and $\lt_\succ(f)=\lc_\succ(f)\lm_\succ(f)$. 

\subsection{Primary Decomposition}

First, we recall the four conditions for minimal primary decompositions. 

\begin{Definition}[Primary Decomposition Definition 4.1.1, \cite{greuel2002singular}]
	For a proper ideal $I$ of $R$, a finite set of ideals $\mathcal{Q}=\{Q_1\ldots,Q_r\}$ is called {\em a primary decomposition} of $I$ if it satisfies
	\begin{enumerate}\setlength{\leftskip}{1cm}
		\item[(PD-1)] $Q_1,\ldots,Q_{r-1}$ and $Q_r$ are primary ideals of $R$, 
		\item[(PD-2)] $I=Q_1\cap \cdots \cap Q_r$.
	\end{enumerate}
	 A primary decomposition $\mathcal{Q}=\{Q_1\ldots,Q_r\}$ of $I$ is said to be {\em minimal} or {\em irredundant} if it satisfies
	
	\begin{enumerate}\setcounter{enumi}{2}\setlength{\leftskip}{1cm}
 		\item[(M-1)] $\bigcap_{j\neq i} Q_j\not \subset Q_i$ for all $i$,
		\item[(M-2)] $\sqrt{Q_i}\neq \sqrt{Q_j}$ for all $i\neq j$.
	\end{enumerate}
	Each $Q_i$ is called {\em a primary component} of $I$ for a minimal primary decomposition $\mathcal{Q}=\{Q_1\ldots,Q_r\}$ of $I$.  
\end{Definition}

\begin{Example} \label{ex1}
	For $I=\langle (x_1+a_1)x_1^2,(x_1+a_1)x_2 \rangle$ in $\mathbb{Q}[a_1,x_1,x_2]$, $\{\langle x_1+a_1\rangle,\langle x_1^2,x_2\rangle\}$ is a minimal primary decomposition of $I$, i.e., 
	\[
	I=\langle x_1+a_1\rangle\cap \langle x_1^2,x_2\rangle.
	\] 
\end{Example}

\begin{Remark}
	For $\alpha\in K^m$, $\{\varphi_\alpha (Q_1)\ldots,\varphi_\alpha (Q_r)\}$ is not always a minimal primary decomposition of $\varphi_\alpha(I)$ even if $\{Q_1\ldots,Q_r\}$ is a minimal primary decomposition of $I$. For instance, in Example \ref{ex1}, $I=\langle x_1+a_1\rangle\cap \langle x_1^2,x_2\rangle$ is a minimal primary decomposition of $I$ but $\varphi_0 (\mathcal{Q})=\{\langle x_1\rangle,\langle x_1^2,x_2 \rangle\}$ is not that of $\varphi_0(I)=\langle x_1^3,x_1x_2\rangle=\langle x_1^3 \rangle\cap \langle x_1,x_2 \rangle$. On the other hand, for any $a\neq 0$ in $\mathbb{Q}$, $\varphi_a (\mathcal{Q})=\{\langle x_1+a\rangle,\langle x_1^2,x_2 \rangle\}$ is a minimal primary decomposition of $\varphi_a(I)=\langle (x_1+a_1)x_1^2,(x_1+a_1)x_2\rangle$.
\end{Remark}

We introduce a maximal independent set, which is useful to transform a positive dimensional ideal into a zero-dimensional one. 

\begin{Definition}[Maximal Independent Set \cite{greuel2002singular}, Definition 3.5.3]
    Let $I$ be a proper ideal of $K[A,X]$. A subset $U$ of $A\cup X$ is called a {\em maximal independent set} (MIS in short) of $I$ if $I\cap K[U]=\{0\}$ and the cardinality of $U$ is equal to the Krull-dimension of $I$. 
\end{Definition}

\begin{Example}
    $U=\{a_1,x_2\}$ is a maximal independent set of $I=\langle (x_1+a_1)x_1^2,(x_1+a_1)x_2 \rangle$ in $K[a_1,x_1,x_2]$. 
\end{Example}

We remark that $I^c=IK(U)[(X\cup A)\setminus U]$ is a $0$-dimensional ideal of $K(U)[(X\cup A)\setminus U]$ and $I^c$ is called the {\em contraction} of $I$ with respect to $U$.

\subsection{Comprehensive Gr\"obner System}

Second, we review the basic properties of comprehensive Gr\"obner system (CGS in short). The notion ``locally closed set'' below is important in {\em covering} of $K^m$.   

\begin{Definition}[Localy Closed Set \cite{Brunat}, Section 2] \label{def:ACS}
    A subset $C$ of $K^m$ is called a {\em locally closed set} if there exists ideals $J_1$ and $J_2$ of $K[A]$ such that $C=V(J_1)\setminus V(J_2)$. 
\end{Definition}

\begin{Remark}
    Since $V_K(\langle 1\rangle)=\emptyset$ and $V_K(\langle 0\rangle)=K^m$, both $\emptyset=V_K(\langle 1\rangle)\setminus V_K(\langle 1\rangle)$ and $K^m=V_K(\langle 0\rangle)\setminus V_K(\langle 1\rangle)$ are locally closed sets.
\end{Remark}

\begin{Remark}
    For proper ideals $J_1$ and $J_2$ of $K[A]$, $V_{K}(J_1)\setminus V_{K}(J_2)$
 may be the empty set though $V_{\mathbb{C}}(J_1)\setminus V_{\mathbb{C}}(J_2)$ is not. For example, let $J_1=\langle a^4-1\rangle$ and $J=\langle a^2-1\rangle$ of $\mathbb{Q}[x,a]$ then $V_{\mathbb{C}}(J_1)\setminus V_{\mathbb{C}}(J_2)=\{\pm i,\pm 1\}\setminus \{\pm 1\}=\{\pm i\}$ but $V_{\mathbb{Q}}(J_1)\setminus V_{\mathbb{Q}}(J_2)=\{\pm 1\}\setminus \{\pm 1\}=\emptyset$.
\end{Remark}

\begin{Example}
    $V(\langle a_1\rangle)\setminus V(\langle a_2\rangle)=\{(0, b) | b\not = 0 \in \mathbb{Q}\}$ is a locally closed set of $\mathbb{Q}^2$ for $A=\{a_1,a_2\}$. 
\end{Example}

\begin{Remark}
    In Definition \ref{def:ACS}, we may assume that $J_1\subset J_2$ since $V(J_1)\setminus V(J_2)=V(J_1)\setminus (V(J_1)\cap V(J_2))=V(J_1)\setminus V(J_1+J_2)$. 
\end{Remark}

For a given parametric ideal $I$ of $K[A,X]$, we divide $K^m$ into several locally closed sets such that for each locally closed set $C$, the reduced Gr\"obner basis of $\varphi_\alpha (I)$ is stable, i.e., there exists a finite subset $G$ of $K[A,X]$ such that $\varphi_\alpha(G)$ is the reduced Gr\"obner basis of $\varphi_\alpha (I)$, for any $\alpha \in C$. We call the set of pairs of a locally closed set and the (reduced) Gr\"obner basis {\em a comprehensive Gr\"obner system} of $I$. 

\begin{Definition}[Comprehensive Gr\"obner System \cite{Suzuki-Sato}, Definition 1] \label{def:CGS}
	Let $F$ be a finite set of polynomials in $K[A,X]$, $C_1,\ldots,C_l$ locally closed sets of $K^m$ with $\bigcup_{i=1}^l C_i = K^m$ and $G_1,\ldots,G_l$ finite subsets of $K[A,X]$ where $l$ is a positive integer. A finite set $\mathcal{G}=\{(C_1,G_1),$ $\ldots$ $,(C_l,G_l)\}$ is called a {\em comprehensive Gr\"obner system} of $F$ if $\varphi_\alpha(G_i)$ is a Gr\"obner basis of the ideal $\langle \varphi_\alpha(F) \rangle$ for any $\alpha\in C_i$ and each $i=1,\ldots,l$. 
\end{Definition}

\begin{Example}
    Fix the lexicographic ordering $x_1\succ x_2$. For $A=\{a_1\}$, $X=\{x_1,x_2\}$ and $F=\{a_1x_1^2+x_2,x_2^2\}$ in $\mathbb{C}[a_1,x_1,x_2]$, 
	\[
	\mathcal{G}=\{(\mathbb{C}\setminus V(\langle a_1\rangle),\{a_1x_1^2+x_2,x_2^2\}),(V(\langle a_1\rangle),\{x_2\})\}
	\]
	is a comprehensive Gr\"obner System of $F$ with respect to $\succ$. Here, $V(\langle 0\rangle)\setminus V(\langle a_1\rangle)=\mathbb{C}\setminus \{0\}$ and $V(\langle a_1\rangle)=\{0\}$. 
\end{Example}

Algorithms for CGS have been studied in \cite{Kapur,Montes,Suzuki-Sato,Weispfenning}. Suzuki-Sato Algorithm \cite{Suzuki-Sato} uses the following proposition to compute CGSs efficiently. 

\begin{Proposition}[\cite{Suzuki-Sato}, Lemma 2.2] \label{prop:generalCGS}
Let $G$ be a Gr\"obner basis of an ideal $\langle F\rangle$ in $K[A,X]$ with respect to a block ordering $X\succ\succ A$. If $\varphi_\alpha(\lc
_\succ (g)) \neq 0$ for each $g \in G\setminus K[A]$, then $\varphi_\alpha (G)$ is a Gr\"obner basis of $\langle \varphi_\alpha(F)\rangle$ in $K[X]$ with respect to $\succ_X$ for any $\alpha\in V(\langle G\cap K[A]\rangle)$. 
\end{Proposition}

In order to compute a parametic primary decomposition of $I$ effectively, we first compute primary decompsition of $I\cap K[A]$, which we call {\em the condition ideal of $I$}.

\begin{Definition}
    For an ideal $I$ of $K[A,X]$, we call $I\cap K[A]$ the condition ideal of $I$. If the condition ideal of $I$ is a prime (resp. a radical) ideal then we say that $I$ has a {\em prime} (resp. a {\em radical}) condition ideal. In particular, if the condition ideal of $I$ is zero then we say that $I$ has a {\em generic} condition ideal.
\end{Definition}

\begin{Example}
    $I_1=\langle ax^2,a^2+1\rangle$ has a prime condition ideal $I_1\cap \mathbb{Q}[a]=\langle a^2+1\rangle$. $I_2=\langle x^2+a,y^2\rangle$ has a generic condition ideal $I_1\cap \mathbb{Q}[a]=\langle 0\rangle$.
\end{Example}

\begin{Remark}
    Since $\varphi_\alpha(I)=K[X]$ for any $\alpha\in K^m\setminus V(I\cap K[A])$, it is sufficient to consider only the value $\alpha$ in $V(I\cap K[A])$. Also, we may assume that the condition ideal is a proper radical ideal as $V(I\cap K[A])=V(\sqrt{I\cap K[A]})$. 
\end{Remark}

The following lemma states that the radical of the condition ideal does not change after we replace $I$ by $I+\sqrt{I\cap K[A]]}$.

\begin{Lemma}
    Let $I$ be a proper ideal of $K[A,X]$ and $I^\prime=I+\sqrt{I\cap K[A]}$. Then, the condition ideal of $I^\prime$ is equal to $\sqrt{I\cap K[A]}$. In particular, $I^\prime$ has a radical condition ideal.
\end{Lemma}
\begin{proof}
    Since 
    \begin{align*}
        (I+\sqrt{I\cap K[A]})\cap K[A]&\subset (\sqrt{I}+\sqrt{I\cap K[A]})\cap K[A]\\
        &\subset \sqrt{I}\cap K[A]\\
        &=\sqrt{I\cap K[A]}\\
        &\subset (I+\sqrt{I\cap K[A]})\cap K[A], 
    \end{align*}
    we obtain $I^\prime\cap K[A]=(I+\sqrt{I\cap K[A]})\cap K[A]=\sqrt{I\cap K[A]}$.
\end{proof}

For recurrent computations, we consider an ideal $J$ of $K[A]$ containing $I\cap K[A]$ properly. If $I\cap K[A]$ is radical then varieties of $I\cap K[A]$ and $J$ are different over $\mathbb{C}$ as follows. 

\begin{Lemma}\label{lem:procon}
    Let $J_1$ and $J_2$ be ideals of $K[A]$. If $J_1$ is a proper radical ideal and $J_2\supsetneq J_1$ then $V_{\mathbb{C}}(J_1)\setminus V_{\mathbb{C}}(J_2)\neq \emptyset$.
\end{Lemma}
\begin{proof}
    As $J_2\supsetneq J_1$, $J_2\not \subset J_1=\sqrt{J_1}$ and then $\sqrt{J_2}\not \subset \sqrt{J_1}$ i.e. $V_{\mathbb{C}}(J_1)\not \subset V_{\mathbb{C}}(J_2)$ by Hilbert's Nullstellensatz. Since $J_1$ is proper, $V_{\mathbb{C}}(J_1)\neq \emptyset$ and $V_{\mathbb{C}}(J_1)\setminus V_{\mathbb{C}}(J_2)$ is not empty.
\end{proof}

Proposition \ref{prop:generalCGS} can be rewritten using the language of locally closed sets as follows. Corollary \ref{cor:generalCGS} is very useful to compute stable pairs of ideal operations which we discuss in Section 3.

\begin{Corollary}\label{cor:generalCGS}
    Let $I$ be a proper ideal with a radical condition ideal. Let $G$ be the reduced Gr\"obner basis of $I$ in $K[A,X]$ with respect to a block ordering $X\succ\succ A$. Then, there exists an ideal $J$ of $K[A]$ such that $J\supsetneq I\cap K[A]$ and $\varphi_\alpha(G)$ is a Gr\"obner basis of $\varphi_\alpha(I)$ with respect to $\succ_X$ for any $\alpha\in V(I\cap K[A])\setminus V(J)$. In particular, $V_{\mathbb{C}}(I\cap K[A])\setminus V_{\mathbb{C}}(J)\neq \emptyset$. 
\end{Corollary}
\begin{proof}
     Let $J=\langle \prod_{g\in G\setminus K[A]}\lc_{\succ}(g)\rangle+I\cap K[A]$. As $G$ is the reduced Gr\"obner basis of $I$ with respect to the block ordering $X\succ\succ A$ , $G\cap K[A]$ is the reduced Gr\"obner basis of $I\cap K[A]$ with respect to $\succ_A$. Thus, for  each $g\in G$, $\lc_\succ(g)$ is not divided by any $\lt_\succ(h)\in G\cap K[A]$ i.e. $\lc_\succ(g)\not \in I\cap K[A]$. Thus, $J\supsetneq I\cap K[A]$.

    For $\alpha\in V(I\cap K[A])\setminus V(J)$, $\varphi_\alpha(\lc
_\succ (g)) \neq 0$ for each $g \in G\setminus K[A]$. Thus, from Proposition \ref{prop:generalCGS}, $\varphi_\alpha(G)$ is a Gr\"obner basis of $\langle \varphi_\alpha(G)\rangle$ with respect to $\succ_X$ for any $\alpha\in V_{\mathbb{C}}(G\cap K[A])\setminus V_{\mathbb{C}}(J)$. Since $\varphi_{\alpha}(\langle G\rangle)=\langle \varphi_\alpha(G)\rangle$ for any $\alpha\in K^m$ and $G\cap K[A]$ is a Gr\"obner basis of $I\cap K[A]$ with respect to $\succ_{A}$, $G$ is also a Gr\"obner basis of $\varphi_\alpha(I)=\varphi_{\alpha}(\langle G\rangle)$ with respect to $\succ_X$ for any $\alpha\in V(I\cap K[A])\setminus V(J)$. By Lemma \ref{lem:procon}, $V_{\mathbb{C}}(I\cap K[A])\setminus V_{\mathbb{C}}(J)$ is not empty.
\end{proof}

We see some properties of locally closed sets in the following. First, we recall that the intersection of locally closed sets is also a locally closed set. Similarly, the set difference of two locally closed sets is a union of two locally closed sets. 

\begin{Lemma}[\cite{Yokoyama2006}, Remark 1] \label{lem:union-of-lcs}
    For locally closed sets $C_1$ and $C_2$, $C_1\cap C_2$ is a locally closed set. Also, $C_1\setminus C_2$ is a union of two locally closed sets.
\end{Lemma}

In general, for a comprehensive Gr\"obner system $\{(G_1,C_1),\ldots,(G_l,C_l)\}$ of $(\mathcal{F},(I_1,\ldots,I_r))$, there can be $C_i$ and $C_j$ which have a common stable basis $G$, i.e., $(G,C_i)$ and $(G,C_j)$ are stable pairs of $(\mathcal{F},(I_1,\ldots,I_r))$. In this case, we can combine the two pairs as a new stable pair $(G,C_i\cup C_j)$. Thus, we consider a union of locally closed sets as follows.  

\begin{Definition}[Constructible Set \cite{Brunat}, Section 3]
    We call a finite union of locally closed sets {\em a constructible set}. 
\end{Definition}

We recall that the set of all constructible sets is closed in the following set operations. 

\begin{Lemma}[\cite{Brunat}, Lemma 3.1] \label{lem:cs}
    The set of all constructible sets is closed under the union, intersection, and difference. In other words, $C_1\cup C_2$, $C_1\cap C_2$, and $C_1\setminus C_2$ are all constructible sets for any constructible sets $C_1$ and $C_2$. 
\end{Lemma}

\section{Stability of Ideal Operations}
Here, by ideal operations, we call operations that create new ideals from an ideal or ideals, such as for ideals $I_1,\ldots$, $I_{r-1}$ and $I_r$; 
    \begin{itemize}
        \item[(O1)] the sum $I_1+ \cdots +I_r$
        \item[(O2)] the product $I_1\cdot I_2\cdot \cdots \cdot I_r$ \quad ($\langle f\rangle\cdot I$ is abbreviated as $f\cdot I$ for $f\in R$)
        \item[(O3)] the intersection $I_1\cap \cdots \cap I_r$
        \item[(O4)] the quotient $I_1:I_2=\{f\in R\mid f\cdot I_2\subset I_1\}$
        \item[(O5)] the saturation $I_1:I_2^\infty=\bigcup_{k=1}^\infty (I_1:I_2^k)$
		\item[(O6)] the radical $\sqrt{I}=\{f\in R\mid f^m \in I\text{ for some $m\in \mathbb{N}$}\}$
		\item[(O7)] the equidimensional hull $\bigcap_{Q\in \mathcal{Q},\dim Q=\dim I} Q$, where $\mathcal{Q}$ is a minimal primary decomposition of $I$
		\item[(O8)] the localization $I^{ec}=K(U)IK[X\setminus U]\cap K[X]$ with respect to $U$, where $U$ is a maximal independent set of $I$.
    \end{itemize}
Throughout this paper, the input for each ideal of ideal operations is its generating system or its Gr\"obner basis. In this section, we recall the {\em stability} of ideal operations introduced 
by \cite{Yokoyama2006} and discuss a comprehensive system for a parametric primary decomposition. 

\begin{Definition}[Stability of Ideal Operation]
For ideals $I_1,\ldots,I_r$ of $K[A,X]$, an ideal operation $\mathcal{F}$ and a locally closed set $C$ of $K^m$, we say that $(\mathcal{F},(I_1,\ldots,I_r))$ is {\em stable} on $C$ if 
there exists a finite subset $G$ of $K[A,X]$ such that $\varphi_\alpha(G)$ is a Gr\"obner basis of the ideal $\mathcal{F}(\varphi_\alpha(I_1),\ldots,\varphi_\alpha(I_r))$ of $K[X]$ for all $\alpha\in C$. We call such $C$, $G$ and $(S,G)$ {\em a stable cell, a stable basis and a stable pair of $(\mathcal{F},(I_1,\ldots,I_r))$} respectively.
\end{Definition}

\subsection{Comprehensive Gr\"obner system for Ideal Operations}

First, we introduce a notion of comprehensive system ({\em cell decomposition} in \cite{Yokoyama2006}) of an ideal operation and input ideals (cf. Definition \ref{def:CGS}).  

\begin{Definition} 
    Let $C_1,\ldots,C_l$ be locally closed sets with $\bigcup_{i=1}^l C_i=K^m$, $G_1,\ldots,G_l$ finite subsets of $K[A,X]$ and $(\mathcal{F},(I_1,\ldots,I_r))$ be a pair of an ideal operation and ideals of $K[A,X]$. We say that $\{(C_1,G_1),\ldots,(C_l,G_l)\}$ is a  {\em comprehensive Gr\"obner system} of $(\mathcal{F},(I_1,\ldots,I_r))$ if $(C_i,G_i)$ is a stable pair of $(\mathcal{F},(I_1,\ldots,I_r))$ for each $i=1,\ldots,l$.
\end{Definition}

\begin{Example}
	Let $I_1=\langle x_1+a_1\rangle$ and $I_2=\langle x_1^2,x_2\rangle$. Fix the lexicographic ordering $x_1\succ x_2$. We consider ideal intersection ${\tt Int}$ as an ideal operation, that is, ${\tt Int}(I_1,I_2)=I_1\cap I_2$. For $\alpha \neq 0$,
    \[  
    \varphi_\alpha (I_1)\cap \varphi_\alpha (I_2)=\langle (x_1+\alpha)x_1^2,(x_1+\alpha)x_2 \rangle.
    \]
    On the other hand, for $\alpha=0$, 
    \[  
    \varphi_\alpha (I_1)\cap \varphi_\alpha (I_2)=\langle x_1^2,x_1x_2 \rangle.
    \]
    Thus, $\{(\mathbb{Q}\setminus V(a_1),\{(x_1+\alpha)x_1^2,(x_1+\alpha)x_2\}),(V(a_1),\{x_1^2,x_1x_2\})\}$ is a comprehensive Gr\"obner system of $({\tt Int},(I_1,I_2))$ with respect to $\succ$. 
\end{Example}

Ideal operations which can be computed directly from Gr\"obner bases have comprehensive systems as follows. Algorithms for those ideal operations are explained in \cite{greuel2002singular}. 

\begin{Proposition}[\cite{Yokoyama2006}, Section 2] \label{prop:ideal_op}
    All of the Ideal operations $(O1),\ldots,(O8)$ have comprehensive Gr\"obner systems for any input ideals of $K[A,X]$. 
\end{Proposition}
\begin{proof}
    Since we can compute those operations by using Gr\"obner bases, we can utilize comprehensive Gr\"obner systems to find comprehensive systems of those those operations and input ideals. 
\end{proof}

\begin{Remark}\label{rem:o1o2}
    One can compute directly comprehensive system of operations $(O1)$ and $(O2)$ since $\varphi_\alpha(I_1+I_2)=\varphi_\alpha(I_1)+\varphi_\alpha(I_2)$ and $\varphi_\alpha(I_1\cdot I_2)=\varphi_\alpha(I_1)\cdot \varphi_\alpha(I_2)$ for any ideals and $\alpha\in K^m$. 
\end{Remark}

If the condition ideal $I\cap K[A]$ is not primary and has a primary decomposition $I\cap K[A]=H_1\cap \cdots \cap H_l$, then one can compute a CPDS $\mathcal{G}_i$ of $I+H_i$ for each $i$ and a CPDS  $\mathcal{G}=\mathcal{G}_1\cup \cdots \cup \mathcal{G}_l$ of $I$. Thus, we may assume that the condition ideal is radical and primary i.e. prime ideal. Also, primary components of $I$ that do not vanish on a zero point $\alpha$ of the condition ideal of $I$ are not needed for a primary decomposition of $\varphi_{\alpha}(I)$ since $\varphi_\alpha(Q)=K[X]$ for such $Q$ and $\alpha$. Thus, we introduce a filtered primary decomposition as follows.

\begin{Definition}
    Let $I$ be a proper ideal of $K[A,X]$ with a prime condition ideal and $\mathcal{Q}=\{Q_1,\ldots,Q_r\}$ a minimal primary decomposition of $I$. Then, we call the set $\mathcal{Q}^\prime=\{Q_i\in \mathcal{Q}\mid \sqrt{Q\cap K[A]}=\sqrt{I\cap K[A]}\}=\{Q_{i_1},\ldots,Q_{i_k}\}$ {\em the filtered primary decomposition of $I$ with respect to $\mathcal{Q}$}.
\end{Definition}

\begin{Example}
    Let $I=\langle x_1^2,a_2x_1x_2,a_1\rangle$. Then, $I\cap K[a_1,a_2]=\langle a_1\rangle$ and $\mathcal{Q}=\{\langle x_1,a_1\rangle,\langle x_1,a_1,a_2\rangle,\langle x_1^2,x_2,a_1\rangle\}$ is a minimal decomposition of $I$. Thus, the filtered primary decomposition of $I$ with respect to $\mathcal{Q}$ is $\{\langle x_1,a_1\rangle,\langle x_1^2,x_2,a_1\rangle\}$ since $\langle x_1,a_1\rangle\cap \mathbb{Q}[a_1,a_2]=\langle a_1\rangle$, $\langle x_1,a_1,a_2\rangle\cap \mathbb{Q}[a_1,a_2]=\langle a_1,a_2\rangle$ and $\langle x_1^2,x_2,a_1\rangle\cap \mathbb{Q}[a_1,a_2]=\langle a_1\rangle$.
\end{Example}

It follows from the following lemma that the filtered primary decomposition has at least one primary component of $I$.

\begin{Lemma}\label{lem:conmin}
    The filtered primary decomposition of $I$is not empty.
\end{Lemma}
\begin{proof}
    Since $I=Q_1\cap \cdots \cap Q_r$ and $I\cap K[A]=(Q_1\cap K[A])\cap \cdots \cap (Q_r\cap K[A]))$, 
    \[
    \sqrt{I\cap K[A])}=\sqrt{(Q_1\cap K[A])}\cup \cdots \cup \sqrt{Q_r\cap K[A]}.
    \]
    As $\sqrt{I\cap K[A]}$ and $\sqrt{Q_i\cap K[A]}$ are prime ideals, there exists $i$ such that $\sqrt{I\cap K[A]}=\sqrt{Q_i\cap K[A]}$. Thus, $\mathcal{Q}^\prime$ is not empty.
\end{proof}

In the filtered primary decomposition $\mathcal{Q}^\prime$ of $I$, it is enough to consider points in $V(I\cap K[A])$ as follows.

\begin{Lemma}\label{lem:filt}
     Let $\mathcal{Q}^\prime$ be the filtered primary decomposition of a proper ideal $I$ of $K[A,X]$ with respect to a minimal primary decomposition $\mathcal{Q}$ of $I$. For any non-empty subset $\mathcal{Q}^{\prime\prime}$ of $\mathcal{Q}^\prime$,
     \[
     V\left(\bigcap_{Q\in \mathcal{Q}^{\prime\prime}}(Q\cap K[A])\right)=V(I\cap K[A]).
     \]
\end{Lemma}
\begin{proof}
    By the definition of the filtered primary decomposition, 
    \[
    V(Q\cap K[A])=V(I\cap K[A])
    \]
    for any $Q\in \mathcal{Q}^\prime$. Thus, 
    \[
    V\left(\bigcap_{Q\in \mathcal{Q}^{\prime\prime}}(Q\cap K[A])\right)=\bigcup_{Q\in \mathcal{Q}^{\prime\prime}}V(Q\cap K[A])=V(I\cap K[A]).
    \]
\end{proof}

When $I$ has a prime condition ideal, we obtain the following lemma.

\begin{Lemma}\label{lem:pa}
    Let $I$ be a proper ideal of $K[A,X]$ with a prime condition ideal and $J_1,\ldots,J_k$ ideals o f $K[A]$. If $J_i\supsetneq I\cap K[A]$ for each $i$, then 
    \[
    \bigcap_{i=1}^k J_i\supsetneq I\cap K[A].
    \]
\end{Lemma}
\begin{proof}
    As $J_i\supsetneq I\cap K[A]$, it is obvious  that $\bigcap_{i=1}^k J_i\supset I\cap K[A]$. 
    For $x_i\in J_i\setminus I\cap K[A]$, the product $x_1\cdots x_k\in \bigcap_{i=1}^kJ_i\setminus I\cap K[A]$ since $I\cap K[A]$ is a prime ideal. Thus, $\bigcap_{i=1}^k J_i\supsetneq I\cap K[A]$.
\end{proof}

The condition (PD-2) of primary decompositions holds on  a locally closed set as follows.

\begin{Proposition} \label{prop:genpd}
   Let $I$ be a proper ideal of $K[A,X]$ with a prime condition. Also, let $\mathcal{Q}=\{Q_1,\ldots,Q_r\}$ be a minimal primary decomposition of $I$ and $\mathcal{Q}^\prime=\{Q_{i_1},\ldots,Q_{i_k}\}$ the filtered primary decomposition of $I$ with respect to $\mathcal{Q}$. One can compute an ideal $J\supsetneq I\cap K[A]$ such that and 
    \[
    \varphi_\alpha(I)=\varphi_\alpha(Q_{i_1})\cap \cdots \cap\varphi_\alpha(Q_{i_k})
    \]
    for any $\alpha\in V(I\cap K[A])\setminus V(J)$. Also, $V_{\mathbb{C}}(I\cap K[A])\setminus V_{\mathbb{C}}(J)\neq \emptyset$ for such $J$.
\end{Proposition}
\begin{proof}
    Let $T=\{t_1,\ldots,t_{r-1}\}$ be slack variables and 
	\[
	H(Q_1,\ldots,Q_r)=t_1\cdot Q_1+t_2\cdot Q_2+\cdots +\left(1-\sum_{i=1}^{r-1}t_i\right)\cdot Q_r\subset K[T,A,X]
	\]
    which is an ideal of $K[T,A,X]$. Then, 
    \[
    H(Q_1,\ldots,Q_r)\cap K[A,X]=Q_1\cap \cdots \cap Q_r=I 
    \]
    and 
    \begin{align*}
        H(Q_1,\ldots,Q_r)\cap K[A]&=(H(Q_1,\ldots,Q_r)\cap K[A,X])\cap K[A]\\
        &=I\cap K[A]\tag{1}
    \end{align*}
    Also, 
    \begin{align*}
        \varphi_\alpha (H(Q_1,\ldots,Q_r))&=t_1\cdot \varphi_\alpha(Q_1)+t_2\cdot \varphi_\alpha(Q_2)+\cdots +\left(1-\sum_{i=1}^{r-1}t_i\right)\cdot \varphi_\alpha(Q_r)\\
        &=H(T,\varphi_\alpha (Q_1),\ldots,\varphi_\alpha (Q_r))
    \end{align*}
    for any $\alpha\in K^m$ by Remark \ref{rem:o1o2}.
    Thus, 
    \[
    \varphi_\alpha (H(Q_1,\ldots,Q_r))\cap K[X]=\varphi_\alpha (Q_1)\cap \cdots \cap \varphi_\alpha (Q_r)\tag{2}
    \]
    for any $\alpha\in K^m$. 

    Let $G$ be the reduced Gr\"obner basis of $H(Q_1,\ldots,Q_r)$ with respect to a block ordering $\succ$ with $T\succ\succ X\succ\succ A$. By Corollary \ref{cor:generalCGS}, there exists an ideal $J_1$ of $K[A]$ such that $\varphi_\alpha(G)$ is a Gr\"obner basis of $\varphi_\alpha(H(Q_1,\ldots,Q_r))$ with respect to $\succ_X$ for any $\alpha\in V(H(Q_1,\ldots,Q_r)\cap K[A])\setminus V(J_1)\neq\emptyset$. Then, $\varphi_\alpha(G)\cap K[X]$ is a Gr\"obner basis of $\varphi_\alpha (H(Q_1,\ldots,Q_r))\cap K[X]=\varphi_\alpha (Q_1)\cap \cdots \cap \varphi_\alpha (Q_r)$ with respect to $\succ_X$ for any $\alpha\in V(I\cap K[A])\setminus V(J_1)$ by $(2)$. Since $G\cap K[A,X]$ is the reduced Gr\"obner basis of $H(Q_1,\ldots,Q_r)\cap K[A,X]=Q_1\cap \cdots \cap Q_r=I$ with respect to $\succ_{A,X}$, $\varphi_\alpha(G\cap K[A,X])=\varphi_\alpha(G)\cap K[X]$ is a Gr\"obner basis of $\varphi_\alpha(I)$ for any $\alpha\in V(I\cap K[A])\setminus V(J_1)$ i.e. $\varphi_\alpha(I) = \varphi_\alpha(Q_1)\cap \cdots \cap\varphi_\alpha(Q_r)$ holds. Here,  $Q_i\cap K[A]\supsetneq I\cap K[A]$ for $i\not \in \{i_1,\ldots,i_k\}$. Let $J=J_1\cap \bigcap_{i\not \in \{i_1,\ldots,i_k\}} (Q_i\cap K[A])$. Then, $J\supsetneq I\cap K[A]$ by Lemma \ref{lem:pa}. Therefore,
    \[
    \varphi_\alpha(I)=\varphi_\alpha(Q_1)\cap \cdots \cap\varphi_\alpha(Q_r)=\varphi_\alpha(Q_{i_1})\cap \cdots \cap\varphi_\alpha(Q_{i_k})
    \]
    for any $\alpha\in V(I\cap K[A])\setminus V(J)$ since $\varphi_\alpha(Q_i)=K[X]$ for each $i\not \in \{i_1,\ldots,i_k\}$ and $\alpha\in V(I\cap K[A])\setminus V(Q_i\cap K[A])\subset  V(I\cap K[A])\setminus V(J)$. Here, $V_{\mathbb{C}}(I\cap K[A])\setminus V_{\mathbb{C}}(J)\neq \emptyset$ by Lemma \ref{lem:procon}.
\end{proof}

\begin{Example}
Let $I=\langle (x_1+a_1)x_1^2,(x_1+a_1)x_2 \rangle$ in Example~\ref{ex1}. For the minimal primary decomposition $\{Q_1=\langle x_1+a_1,x_2\rangle,Q_2=\langle x_1^2,x_2\rangle \}$, let $F(T)=t_1\cdot \langle x_1+a_1\rangle+(1-t_1)\cdot \langle x_1^2,x_2\rangle$. A comprehensive Gr\"obner system of $F(T)$ with respect to the lexicographic ordering $t\succ x_1\succ x_2$ is $\mathcal{G}=\{(\mathbb{Q}\setminus V_\mathbb{Q}(a_1),\{(x_1+a_1)x_1^2,(x_1+a_1)x_2,x_1^2-a_1^2t_1\}),(V_{\mathbb{Q}}(a_1),\{x_1x_2,x_1^2,(1-t)x_2,tx_1\})\}$. Here, for $G=\langle (x_1+a_1)x_1^2,(x_1+a_1)x_2,x_1^2-a_1^2t_1\rangle$ in the first segment of $\mathcal{G}$, $G^\prime=G\cap \mathbb{Q}[a_1,x_1,x_2])=\{(x_1+a_1)x_1^2,(x_1+a_1)x_2\}$ and $\varphi_\alpha (G^\prime)$ is a Gr\"obner basis of $\varphi_\alpha(I_1)\cap \varphi_\alpha(I_2)$ for $\alpha \in \mathbb{Q}\setminus V_{\mathbb{Q}}(a_1)$ with respect to the lexicographic ordering $x_1\succ x_2$. Thus, for $\alpha\in \mathbb{Q}$,  $\varphi_\alpha (I)=\varphi_\alpha (Q_1)\cap \varphi_\alpha (Q_2)=\langle (x_1+a_1)x_1^2,(x_1+a_1)x_2\rangle$ if and only if $\alpha \in \mathbb{Q}\setminus V_{\mathbb{Q}}(a_1)$.  
\end{Example}
The following lemma states that for the property of ideal inclusion is stable on a constructible set. 

\begin{Lemma}\label{lem:inclusion}
    Let $I_1$ and $I_2$ be ideals of $K[A,X]$ with $I_1\not \subset I_2$. Assume that $I_2$ is a primary ideal with a prime condition ideal. Then  one can compute an ideal $J$ such that $J\supsetneq I_2\cap K[A]$ and 
    \[
    \varphi_\alpha(I_1)\not\subset \varphi_\alpha(I_1)
    \]
    for any $\alpha\in V(I_2\cap K[A])\setminus V(J)$. Also, $V_{\mathbb{C}}(I_2\cap K[A])\setminus V_{\mathbb{C}}(J) \neq \emptyset$ for such $J$. 
\end{Lemma}

\begin{proof}
   Let $G=\{g_1,\ldots,g_r\}$ be the reduced Gr\"obner basis of $I_2$ with respect to a block ordering $X\succ\succ A$. Then there exists an ideal $J_1\supsetneq I_2\cap K[A]$ such that $\varphi_\alpha (G)$ is a Gr\"obner basis of $I_2$ with respect to $\succ_X$ for any $\alpha\in V(I_2\cap K[A])\setminus V(J_1)$. Let $f\in I_1\setminus I_2$ and $c=\prod_{g\in G\setminus K[A]} \lc_\succ (g)$. Since $I_2$ is primary and $c\not \in I\cap K[A]=\sqrt{I\cap K[A]}$, $c^kf\in I_1\setminus I_2$ for any positive number $k$. There exist polynomials $h_1,\ldots,h_r$ and a positive integer $k$ such that 
    \[
    c^kf=h_1g_1+\cdots +h_rg_r+NF_G(c^kf).
    \]
    and $\varphi_\alpha(NF_G(c^kf))=NF_{\varphi_\alpha(G)}(\varphi_\alpha(c^kf))$ for any $\alpha\in V(I_2\cap K[A])\setminus V(J_1)$ since $\{\varphi_\alpha (g_1),\ldots,\varphi_\alpha (g_r)\}$ is a Gr\"obner basis of $\varphi_\alpha(I)$. Let $J_2$ be the ideal generated by coefficients in $K[A]$ of a term of $NF_{G}(c^kf)$ and elements of $I_2\cap K[A]$. As $NF_G(c^kf)\not \in I_2\cap K[A]$, $J_2\supsetneq I_2\cap K[A]$. Therefore, for $J=J_1\cap J_2$, $J\supsetneq I_2\cap K[A]$ by Lemma \ref{lem:pa}. Hence, $NF_{\varphi_\alpha(G)}(\varphi_\alpha(c^kf))=\varphi_\alpha(NF_G(c^kf))\neq 0$ and thus $\varphi_\alpha(I_1) \not \subset \varphi_\alpha(I_2)$ for any $\alpha\in V(I_2\cap K[A])\setminus V(J)$.
\end{proof}

The minimality of a primary decomposition with respect to intersection, i.e., the condition (M-1) holds on  a locally closed set as follows.

\begin{Proposition} \label{prop:int_gen}
    Let $I$ be a proper ideal of $K[A,X]$ with a prime condition ideal and $\mathcal{Q}=\{Q_1,\ldots,Q_r\}$ a minimal primary decomposition of $I$. Let $\mathcal{Q}^\prime=\{Q_{i_1},\ldots,Q_{i_k}\}$ be the filtered primary decomposition of $I$ with respect to $\mathcal{Q}$. One can compute an ideal $J\supsetneq I\cap K[A]$ such that for each $i\in \{i_1,\ldots,i_k\}$
    \[
    \varphi_\alpha\left(\bigcap_{i_j\neq i}Q_{i_j}\right)=\bigcap_{i_j\neq i}\varphi_\alpha(Q_{i_j})\not \subset \varphi_\alpha (Q_i)
    \]
    for any $\alpha\in V(I\cap K[A])\setminus V(J)$. In particular, $V_{\mathbb{C}}(I\cap K[A])\setminus V_{\mathbb{C}}(J)\neq \emptyset$.
\end{Proposition}

\begin{proof}
By Proposition \ref{prop:genpd}, there exists an ideal $J_1\supsetneq I\cap K[A]$
    \[
\varphi_\alpha(I)=\varphi_\alpha(Q_{i_1})\cap \cdots \cap\varphi_\alpha(Q_{i_k})
    \]
    for any $\alpha\in V(I\cap K[A])\setminus V(J_1)$. For each $i\in \{i_1,\ldots,i_k\}$, there exists an ideal $J_2^{(i)}\supsetneq I\cap K[A]$ such that  
     \[
    \varphi_\alpha\left (\bigcap_{i_j\neq i} Q_{j_i}\right)=\bigcap_{i_j\neq i}\varphi_\alpha(Q_{i_j})
    \]
    for any $\alpha\in V(\bigcap_{i_j\neq i} Q_{j_i}\cap K[A])\setminus V(J_2^{(i)})$. By Lemma \ref{lem:filt}, we replace $ V(\bigcap_{i_j\neq i} Q_{j_i}\cap K[A])\setminus V(J_2^{(i)})$ by  $V(I\cap K[A])\setminus V(J_2^{(i)})$. Also, by Lemma \ref{lem:inclusion} for each $i\in \{i_1,\ldots,i_k\}$, there exists an ideal $J_3^{(i)}\supsetneq I\cap K[A]$ such that
        \[
    \varphi_\alpha\left (\bigcap_{i_j\neq i} Q_{j_i}\right)\not \subset \varphi_\alpha (Q_i).
    \]
    for any $\alpha\in V(Q_i\cap K[A])\setminus V(J_3^{(i)})=V(I\cap K[A])\setminus V(J_3^{(i)})$.
     Let $J=J_1\cap \bigcap_{i\in \{i_1,\ldots,i_k\}}J_2^{(i)}\cap \bigcap_{i\in \{i_1,\ldots,i_k\}}J_3^{(i)}$ then $J\supsetneq I\cap K[A]$ by Lemma \ref{lem:pa}. For each $i\in \{i_1,\ldots,i_k\}$,
    \[
    \varphi_\alpha\left(\bigcap_{i_j\neq i}Q_{i_j}\right)=\bigcap_{i_j\neq i}\varphi_\alpha(Q_{i_j})\not \subset \varphi_\alpha (Q_i)
    \]
    for any $\alpha\in V(I\cap K[A])\setminus V(J)$.
\end{proof}

The minimality of a primary decomposition with respect to radical, i.e., the condition (M-2) holds on  a locally closed set as follows.

\begin{Proposition}\label{prop:rad_gen}
    Let $I$ be a proper ideal of $K[A,X]$ and $\mathcal{Q}=\{Q_1,\ldots,Q_r\}$ a minimal primary decomposition of $I$. Let $\mathcal{Q}^\prime=\{Q_{i_1},\ldots,Q_{i_k}\}$ be the filtered primary decomposition of $I$ with respect to $\mathcal{Q}$. For each $i\neq j\in \{i_1,\ldots,i_k\}$, one can compute an ideal $J_{ij}\supsetneq I\cap K[A]$ such that 
    \[
    \sqrt{\varphi_\alpha (Q_{i})}=\varphi_\alpha(\sqrt{Q_{i}})
     \neq  \sqrt{\varphi_\alpha (Q_{j})}=\varphi_\alpha(\sqrt{Q_{j}})
    \]
    for any $\alpha\in V(I\cap K[A])\setminus V(J_{ij})$. In particular, $V_{\mathbb{C}}(I\cap K[A])\setminus V_{\mathbb{C}}(J_{ij})\neq \emptyset$.
\end{Proposition}
\begin{proof}
     Fix a pair $(i,j)$ with $i\neq j$. By Proposition \ref{prop:ideal_op} and Lemma \ref{lem:inclusion}, there exist ideals $J_{1,ij}$, $J_{2,ij}$ and $J_{3,ij}$ properly containing $I\cap K[A]$ such that 

    \[
    \varphi_\alpha\left (\sqrt{Q_i}\right) = \sqrt{\varphi_\alpha (Q_i)}, \quad \varphi_\alpha\left (\sqrt{Q_j}\right) = \sqrt{\varphi_\alpha (Q_j)}
    \]
    for any $\alpha\in V(I\cap K[A])\setminus V(J_{1,ij})$ and 
    \[
    \varphi_\alpha \left(\sqrt{Q_{i}}\right)\not \subset  \varphi_\alpha \left(\sqrt{Q_{j}}\right)
    \]
    for any $\alpha\in V(\sqrt{Q_j}\cap K[A])\setminus V(J_{2,ij})$ and 
    \[
    \varphi_\alpha \left(\sqrt{Q_{j}}\right)\not \subset  \varphi_\alpha \left(\sqrt{Q_{i}}\right)
    \]
    for any $\alpha\in V(\sqrt{Q_i}\cap K[A])\setminus V(J_{3,ij})$. Let $J=\bigcap_{i\neq j\in\{i_1,\ldots,i_k\}}J_{1,ij}\cap (\bigcap_{i\neq j\in\{i_1,\ldots,i_k\}}J_{2,ij} +\bigcap_{i\neq j\in\{i_1,\ldots,i_k\}}J_{3,ij})\cap \bigcap_{i\not \in\{i_1,\ldots,i_k\}} (\sqrt{Q_i}\cap K[A])$. Then, $J\supsetneq I\cap K[A]$ by Lemma \ref{lem:pa}. For each $i\neq j\in \{i_1,\ldots,i_k\}$,
     \[
    \sqrt{\varphi_\alpha (Q_{i})}=\varphi_\alpha(\sqrt{Q_{i}})
     \neq  \sqrt{\varphi_\alpha (Q_{j})}=\varphi_\alpha(\sqrt{Q_{j}})
    \]
    for any $\alpha\in V(I\cap K[A])\setminus V(J)$. By Lemma \ref{lem:procon}, $V_{\mathbb{C}}(I\cap K[A])\setminus V_{\mathbb{C}}(J)\neq \emptyset$.
\end{proof}

Finally, we introduce the notion of {\em Comprehensive Primary Decomposition System} (CPDS in short). 
In this paper, the words ``parametric primary decomposition'' means this CPDS. The existence and computability of CPDS are proved in Section~\ref{sec:4}.

\begin{Definition}[Comprehensive Primary Decomposition System (CPDS)]\label{def:fCPDS}
    Let $I$ be an ideal of $K[A,X]$ and let $C_1,\ldots,C_l$ be constructible sets and $\mathcal{Q}_1,\ldots,\mathcal{Q}_l$ sets of primary ideals. We call $\{(C_1,\mathcal{Q}_1),$ $\ldots,(C_l,\mathcal{Q}_l)\}$ a {\em CPDS} for $I$ over $K$ if (PD-1) and (PD-2) hold, that is, $\varphi_\alpha(\mathcal{Q}_i)$ is a primary decomposition of $\varphi_\alpha(I)$ for any $\alpha\in C_i$ for each $i$. Moreover, if each $\varphi_\alpha(\mathcal{Q}_i)$ is a minimal primary decomposition of $\varphi_\alpha(I)$, then we call it a minimal CPDS of $I$. We also call $\{(C_1,\mathcal{Q}_1),\ldots,(C_l,\mathcal{Q}_l)\}$ a {\em feasible CPDS} for $I$ if (PD-2) holds. Moreover, if each $\varphi_\alpha(\mathcal{Q}_i)$ satisfies the minimality conditions (M-1) and (M-2), then we call it a feasible minimal CPDS of $I$ and each $(C_i,\mathcal{Q}_i)$ a {\em segment} of its CPDS.
\end{Definition}

\section{Main Algorithms}\label{sec:4}

In this section, we describe the main algorithms for computing a feasible CPDS. We consider Gr\"obner bases as inputs of algorithms. 

\subsection{Algorithms for Feasible CPDS}

First, we provide an algorithm to compute feasible CPDS as Algorithm \ref{alg1}. The specific computational flow of the algorithm is described in Example \ref{ex:alg}.

\begin{Theorem} \label{thm:main-alg}
    For a given ideal $I$ of $K[A,X]$ and its Gr\"ober basis, Algorithm \ref{alg1} terminates and returns a feasible CPDS of $I$ over $K$. 
\end{Theorem}

\begin{proof}    
    First, we show its correctness. By Proposition \ref{prop:genpd}, one can compute an ideal $J$ of $K[A]$ properly containing $H_i$ such that $\varphi_\alpha(I)=\varphi_\alpha(Q_{i_1})\cap \cdots \cap \varphi_\alpha(Q_{i_k})$ for any $\alpha\in V(H_i)\setminus V(J)$ for each $H_i$ at each step in the algorithm. Thus, each segment $(V(H_i)\setminus V(J),\mathcal{Q})$ satisfies the condition of a feasible CPDS of $I$. Also, the set of $V(H_i)\setminus V(J)$ covers $K$ since it computes recursively segments by ``CPDS $\gets \text{\{Segment\}} \cup \textsc{feasibleCPDS}(I+J)$''. Thus, the algorithm returns a feasible CPDS of $I$ when it terminates. 
    
    Next, we prove the termination. If the algorithm does not terminate and performs an infinite number of recursions, an infinite ascending chain of ideals $J^{(1)}\subsetneq J^{(2)}\subsetneq \cdots $ in $K[A]$ is created where $J^{(i)}$ is the $J$ at the $i$-th step recursion. However, it contradicts the fact that $K[A]$ is Noetherian. Therefore, the algorithm terminates after a finite number of recursions.
\end{proof}

\begin{algorithm}[H]            
\caption{\textsc{FeasibleComprehensivePrimaryDecompositionSystem (FeasibleCPDS)}}         
\label{alg1}                          
\begin{algorithmic}[1]                
\REQUIRE $I$: an ideal of $K[A,X]$
\ENSURE Comprehensive Primary Decomposition System of $I$
\IF{$I=K[A,X]$}
    \RETURN $\{\}$
\ENDIF
\STATE CPDS $\gets \{\}$
\STATE Compute a minimal primary decomposition $\{H_1,\ldots,H_l\}$ of $\sqrt{I\cap \mathbb{Q}[A]}$ over $K[A]$
\FOR{$i=1$ to $l$}
    \IF{$I+H_i\neq K[A,X]$}
\STATE Compute a minimal primary decomposition of $I+H_i$ as $\mathcal{Q}=\{Q_1,\ldots,Q_r\}$ in $K[A,X]$
\STATE $\mathcal{Q}^\prime\gets \{Q\in \mathcal{Q}\mid \sqrt{Q\cap K[A]}=\sqrt{(I+H_i)\cap K[A]}\}=\{Q_{i_1},\ldots,Q_{i_k}\}$
     \STATE Compute an ideal $J$ of $K[A]$ properly containing $H_i$ such that $\varphi_\alpha(I)=\varphi_\alpha(Q_{i_1})\cap \cdots \cap \varphi_\alpha(Q_{i_k})$ for any $\alpha\in V(H_i)\setminus V(J)$
     \STATE Segment $\gets (V(H_i)\setminus V(J),\mathcal{Q})$
    \STATE CPDS $\gets \text{\{Segment\}} \cup \textsc{feasibleCPDS}(I+J)$
    \ENDIF
\ENDFOR
\RETURN CPDS
\end{algorithmic}
\end{algorithm}

\begin{Example} \label{ex:alg}
    Let $I=\langle x_1^2-a,bx_1x_2 \rangle$ in $\mathbb{Q}[a,b,x_1,x_2]$. Fix the lexicographic ordering $x_1\succ x_2$. Remark that $I$ has the generic condition ideal i.e. $I\cap \mathbb{Q}[a,b]=\{0\}$. Then, first we compute a primary decomposition $\mathcal{Q}=\{Q_1^{(1)}=\langle x_1^2-a,b \rangle,Q_2^{(1)}=\langle x_1,a \rangle,Q_3^{(1)}=\langle x_1^2-a,x_2 \rangle\}$ of $I$ in $\mathbb{Q}[a,b,x_1,x_2]$, i.e.,
		\begin{align*}
			I&=\langle x_1^2-a,b \rangle\cap \langle x_1,a \rangle\cap \langle x_1^2-a,x_2 \rangle.
		\end{align*}  
        Then, for $\alpha\in \mathbb{Q}^2\setminus V(\langle ab\rangle)=\{(c,d)\mid c\neq 0,d\neq 0\}$, 
        \begin{align*}
            \varphi_\alpha (I)&=\varphi_\alpha(Q_1^{(1)})\cap \varphi_\alpha(Q_2^{(1)})\cap \varphi_\alpha(Q_3^{(1)})\\
                              &=\langle x_1^2-a,x_2\rangle.
        \end{align*}
        Next, let $I_2=I+\langle ab \rangle=\langle x_1^2-a,bx_1x_2,ab\rangle$. We compute a minimal primary decomposition $\{\langle a\rangle,\langle b\rangle\}$
        of $I_2\cap \mathbb{Q}[a,b]=\langle ab\rangle$. Then, we recursively apply the algorithm $I+\langle a\rangle$ and $I+\langle b\rangle$.
        \begin{itemize}
            \item Let $I_3=I+\langle a\rangle=\langle x_1^2,bx_1x_2,a \rangle$. We compute a primary decomposition $\mathcal{Q}^{(2)}=\{Q_1^{(2)}=\langle x_1,a \rangle,Q_2^{(2)}=\langle x_1^2,a.b \rangle,Q_3^{(2)}=\langle x_1^2,x_2,a \rangle\}$ of $I_2$ in $\mathbb{Q}[a,b,x_1,x_2]$. For $\alpha\in V(\langle a\rangle)$, 
           \begin{align*}
            \varphi_\alpha (I)&=\varphi_\alpha(Q_1^{(2)})\cap \varphi_\alpha(Q_2^{(2)})\cap \varphi_\alpha(Q_3^{(2)})\\
                              &=\langle x_1\rangle\cap \langle x_1^2,b\rangle\cap \langle x_1^2,x_2\rangle.
        \end{align*}
        \item Let $I_4=I+\langle b\rangle=\langle x_1^2-a,b \rangle$. We compute a primary decomposition $\mathcal{Q}^{(3)}=\{Q_1^{(3)}=\langle x_1^2-a,b\rangle\}$ of $I_3$ in $\mathbb{Q}[a,b,x_1,x_2]$. For $\alpha\in V(\langle b\rangle)$, 
           \begin{align*}
            \varphi_\alpha (I)&=\langle x_1^2-a\rangle.
        \end{align*}
        \end{itemize}
    
        Finally, we obtain a feasible CPDS of $I$ with respect to $\succ$ over $\mathbb{Q}$:
        \begin{align*}
            \{&(\mathbb{Q}^2\setminus V(\langle ab\rangle),\{\langle x_1^2-a,x_2\rangle\}),(V(\langle a\rangle),\{\langle x_1 \rangle,\langle x_1^2,b\rangle,\langle x_1^2,x_2\rangle,\\
            &(V(\langle b\rangle),\{\langle x_1^2-a\rangle\})\})\}.
        \end{align*}
\end{Example}

\begin{Remark}
    When $A$ consists of one parameter $a$, the computed locally closed set are all finite sets except the first one $V(I\cap K[A])\setminus V(J)$, where $V(J)$ is also a finite set. In this case, we do not need to run Algorithm \ref{alg1} recursively since we can compute ordinary primary decomposition of $\varphi_\alpha(I)$ for all $\alpha\in J$. 
\end{Remark}

\subsection{Algorithm for a minimal feasible CPDS}

In order to compute a minimal feasible CPDS, we utilize Propositions~\ref{prop:int_gen} and \ref{prop:rad_gen}. In Algorithm~\ref{alg2}, we compute ideals $J_1$ and $J_2$ of $K[A]$ such that the specialized primary decomposition $\varphi_\alpha(\mathcal{Q})$ for $\alpha \in V(J_1)\setminus V(J_2)$ satisfies all conditions of minimal primary decomposition except primarity. Since we compute such $\mathcal{Q}$ and $J$ recursively as in Algorithm \ref{alg1}, a minimal feasible CPDS is finally obtained. Thus, the following corollary holds. 

\begin{Corollary}
        Algorithm \ref{alg2} returns a minimal feasible CPDS of $I$ over $K$. 
\end{Corollary}
\begin{proof}
    The termination of Algorithm \ref{alg2} can be proved in the similar way of that of Algorithm \ref{alg1}. In each Segment in \ref{alg2}, it satisfies all conditions for a minimal feasible CPDS by Propositions \ref{prop:int_gen} and \ref{prop:rad_gen}. Thus, Algorithm \ref{alg2} returns a minimal feasible CPDS of $I$ over $K$.
\end{proof}

\begin{algorithm}[H]            
\caption{\textsc{MinimalFeasibleCPDS}}         
\label{alg2}                          
\begin{algorithmic}[1]              
\REQUIRE $I$: an ideal of $K[A,X]$
\ENSURE a minimal feasible CPDS of $I$ 
\STATE CPDS $\gets \{\}$
\STATE Compute a minimal primary decomposition of $I$ as $\mathcal{Q}=\{Q_1,\ldots,Q_r\}$ in $K[A,X]$
\STATE Compute a minimal primary decomposition $\{H_1,\ldots,H_l\}$ of $\sqrt{I\cap K[A]}$ over $\mathbb{Q}[A]$
\FOR{$i=1$ to $l$}
    \IF{$I+H_i\neq K[A,X]$}
\STATE Compute a minimal primary decomposition of $I+H_i$ as $\mathcal{Q}=\{Q_1,\ldots,Q_r\}$ in $K[A,X]$
\STATE $\mathcal{Q}^\prime\gets \{Q\in \mathcal{Q}\mid \sqrt{Q\cap K[A]}=\sqrt{(I+H_i)\cap K[A]}\}=\{Q_{i_1},\ldots,Q_{i_k}\}$
\STATE Renumber the indices $i_1,\ldots,i_k$ as $1,\ldots,k$
     \STATE Compute an ideal $J$ of $K[A]$ properly containing $H_i$ such that \begin{itemize}
    \item $\varphi_\alpha(I)=\varphi_\alpha(Q_1)\cap \cdots \cap \varphi_\alpha(Q_k)$ 
    \item $\varphi_\alpha\left(\bigcap_{j\neq i}Q_j\right)=\bigcap_{j\neq i}\varphi_\alpha(Q_j)\not \subset \varphi_\alpha (Q_i)$ 
    \item $\sqrt{\varphi_\alpha (Q_{i})}=\varphi_\alpha(\sqrt{Q_{i}})
     \neq  \sqrt{\varphi_\alpha (Q_{j})}=\varphi_\alpha(\sqrt{Q_{j}})$ 
\end{itemize} for any $\alpha\in V(H_i)\setminus V(J)$
\STATE Segment $\gets (V(H_i)\setminus V(J),\mathcal{Q})$
    \STATE CPDS $\gets \text{\{Segment\}} \cup \textsc{MinimalFeasibleCPDS}(I+J)$
    \ENDIF
\ENDFOR
\RETURN CPDS
\end{algorithmic}
\end{algorithm}

\section{Primarity of Parametric Primary Ideal}\label{sec5}

In this section, we discuss the primarity of a parametric primary ideal evaluated at each parameter value. In general, for a primary ideal $Q$ of $K[A,X]$, the specialized ideal $\varphi_\alpha(Q)$ is not always a primary ideal e.g. $\varphi_{1}(\langle x_1^2-a_1\rangle)=\langle x_1^2-1\rangle=\langle x+1\rangle\cap \langle x-1\rangle$ even though $\langle x_1^2-a_1\rangle$ is a prime ideal in $K[A,X]$. However, for ``almost'' rational numbers $\alpha$, $\varphi_\alpha(Q)$ is still a primary ideal. 

\subsection{Primarity of Parametric Ideal} 

First, we recall a notion of minimal polynomial as follows.

\begin{Definition}[Minimal Polynomial and Generic Position \cite{Noro2004}, Definitions 2.1, 2.2]
Let $I$ be a zero-dimensional ideal of $K[X]$. For a polynomial $f(X)$ of $K[X]$, 
the {\em minimal polynomial} $m_f(t)$ modulo $I$ is defined as the monic, univariate polynomial over $K$ having the smallest degree among all univariate polynomials $h$ such that $h(f)\in I$. A polynomial $g(X)\in K[X]$ is said to be {\em in generic position} with respect to $I$ if $\deg(\sqrt{m_g(t)})=\dim_K(K[X]/\sqrt{I})$ for the minimal polynomial $m_g$ with respect to $I$ where $\sqrt{m_g(t)}$ is the square free part of $m_g(t)$. 
\end{Definition}

If $g(X)$ is in generic position with respect to $I$, then the factorization of the minimal polynomial of $g(X)$ gives the primary decomposition of $I$ as follows. In general, there exists a polynomial 
which is in generic position with respect to $I$ of $K[X]$ if $K$ is a perfect field.

\begin{Lemma}[\cite{Noro2004}, Proposition 2.3] \label{lem:fact}
    Let $I$ be a zero-dimensional ideal of $K[X]$, and suppose that a polynomial $g(X)$ is in generic position with respect to $I$ and that $m_g$ is the minimal polynomial of $g(X)$ with respect to $J$. Moreover, suppose that 
    \[
    m_g(t)=m_1(t)^{e_1}\cdots m_r(t)^{e_r}
    \]
    is the irreducible factorization of $m_g$ over $K$. Then $P_i=I+\langle m_i(g)\rangle$ is a prime divisor for each $m_i$, and $\sqrt{J}=\bigcap_{i=1}^r P_i$ is the prime decomposition of $I$. 
\end{Lemma}

\begin{Definition}[Local Dimension]
    Let $I$ be an ideal of $K[A,X]$, $C$ a constructible set and $d$ an integer. Then, we say that $I$ is {\em local $d$-dimensional} with respect to $C$ if $\varphi_\alpha(I)$ is $d$-dimensional as an ideal of $K[X]$ for any $\alpha\in C$. 
\end{Definition}

\begin{Example}
    Let $I=\langle x_1^2,x_1x_2,ax_2^2\rangle\subset \mathbb{Q}[a,x_1,x_2]$. The local dimension of $I$ is $1$ with respect to $V(a)$ is $1$ since $\varphi_0 (I)=\langle x_1^2,x_1x_2\rangle\subset \mathbb{Q}[x_1,x_2]$. On the other hand, the local dimension of $I$ is $0$ with respect to $\mathbb{Q}\setminus V(a)$ is $0$ since $\varphi_\alpha (I)=\langle x_1^2,x_1x_2,x_2^2\rangle\subset \mathbb{Q}[x_1,x_2]$ for any $\alpha\neq 0$. 
\end{Example}

\begin{Remark}
    Let $I$ be a parametric ideal of $K[A,X]$. For a comprehensive Gr\"obner system $\mathcal{G}=\{(C_1,G_1),$ $\ldots$ $,(C_l,G_l)\}$, the local dimension of $G_i$ with respect to $C_i$ can be defined for each $i$ since the dimension of the ideal is uniquely determined from its Gr\"obner basis. 
\end{Remark}

\begin{Definition}[Local Maximal Independent Set]
    Let $I$ be an ideal of $K[A,X]$, $C$ a constructible set and $U\subset X$. Then, we say that $U$ is a {\em local maximal independent set} {\em (local MIS, in short)} with respect to $C$ if $U$ is a maximal independent set of $\varphi_\alpha(I)$ for any $\alpha\in C$. 
\end{Definition}

\begin{Example}
    Let $I=\langle x^2,xy,ay^2\rangle\subset \mathbb{Q}[a,x,y]$. Then $U=\{y\}$ is a local maximal independent set with respect to $V(a)$. 
\end{Example}

\begin{Definition}[Local Minimal Polynomial]
    Let $I$ be a local zero-dimensional ideal $I$ with respect to a constructible set $C$ and $G$ be a stable Gr\"obner basis with respect to $C$. For a polynomial $f(A,X)$ of $K[A,X]$ and $m_f(A,t)\in K[A,t]$, $m_f(A,t)$ is called the {\em local minimal polynomial} $m_f(A,t)$ modulo $I$ and $C$ if $m_f(\alpha,t)$ is the minimal polynomial modulo $\varphi_\alpha(I)$ for any $\alpha\in C$. 
    A polynomial $g(A,X)\in K[A,X]$ is said to be {\em in local generic position} with respect to $I$ and $C$ if $\deg(\sqrt{m_g(\alpha,t)})=\dim_K(K[X]/\sqrt{\varphi_\alpha(I)})$ for the local minimal polynomial $m_g$ modulo $I$ for any $\alpha\in C$. 
\end{Definition}

\begin{Example}
    Let $I=\langle x^2,xy,ay^2+a^2y\rangle\subset \mathbb{Q}[a,x,y]$. Then $m_y(a,y)=y^2+ay$ is a local minimal polynomial with respect to $I$ and $\mathbb{Q}\setminus V(a)$ since $\sqrt{\varphi_\alpha(I)}=\langle x,y^2+ay\rangle$ for $\alpha\neq 0$. As $\deg(m_y(\alpha,t))=\dim_K(K[X]/\sqrt{\varphi_\alpha(I)})=2$ for $\alpha\neq 0$, it follows that $y^2+ay$ is in local generic position with respect to $I$ and $\mathbb{Q}\setminus V(a)$. 
\end{Example}

\subsection{Hilbert's Irreducibility Theorem}

To discuss the irreducibility of the minimal polynomial with parameters, we use Hilbert's irreducibility theorem. We recall the notion of Zariski dense. For detailed properties of Zariski topology, see \cite{Hartshorne}. 

\begin{Definition}[Zariski Dense Set]
    Let $(K^m,\mathbb{O})$ be the Zariski topology. Then, a subset $\mathcal{O}$ is said to be 
    {\em Zariski dense} if the closure of $\mathcal{O}$ equals to $K^m$. 
\end{Definition}

\begin{Remark}
    Any Zariski open set is Zariski dense, but the converse is not necessarily true. For example, letting $\pi:\mathbb{Q}^2\to \mathbb{Q}$ be the projection map by $\pi(a_1,x_1)=a_1$, the image of the variety $\pi(V_{\mathbb{Q}}(x_1^2-a_1))$ is Zariski-dense but not Zariski-open in $\mathbb{Q}$. However, in this case, it is also dense for the ordinary (Euclidean) topology (see Lemma \ref{lem:dense}). 
\end{Remark}

We introduce the Hilbert's irreducibility theorem, which was conceived by David Hilbert in 1892. 

\begin{Theorem}[\cite{Corvaja2016}, Theorem 4.1.2] \label{Thm:HIT}
    Let $f(A,X)$ be an irreducible polynomial of degree $\ge 1$ in $\mathbb{Q}[A,X]$. Then, there exist infinitely many points $\alpha\in \mathbb{Q}^m$ such that the specialized polynomial $f(\alpha,X)$ is irreducible in $\mathbb{Q}[X]$. Moreover, there exists a Zariski dense set $\mathcal{O}$ of $\mathbb{Q}^m$ such that $f(\alpha,X)$ is irreducible in $\mathbb{Q}[X]$ for any $\alpha\in \mathcal{O}$. 
\end{Theorem}

\begin{Example}
    Let $f(a_1,x_1)=x_1^2-a_1$ in $\mathbb{Q}[a_1,x_1]$. Then, $\mathcal{O}=\{\alpha\in \mathbb{Q}\mid \text{$\alpha$ is not a positive square number}\}$ is a Zariski-dense set of $\mathbb{Q}$ and $f$ is irreducible over $\mathbb{Q}$ for any $\alpha\in \mathcal{O}$. 
\end{Example}

\begin{Remark}
In \cite{Yokoyama2006} a systematic method for 
factorization of polynomials with parametric coefficients was proposed by introducing the notion of decomposition ideal. In its continuation \cite{Yokoyama2022}, if an irreducible polynomial $f(a,X)$ with a single parameter $a$ over $\mathbb{Q}$ 
is absolutely irreducible over $\mathbb{Q}(a)$, the parameter values making $f$ reducible can be completely classified. See Example \ref{ex:hilb}. 
\end{Remark}

\begin{Remark}
    A field satisfying the property of Theorem \ref{Thm:HIT} is said to be {\em Hilbertian}. A finite extension of a Hilbertian field is also Hilbertian (see Chapter 9 of \cite{Lang}). Thus, for an irreducible polynomial $f$ over $\mathbb{Q}$ and its root $\kappa$, we can apply the Hilbert's irreducible theorem to parametric polynomials over $\mathbb{Q}(\kappa)$. Also, the field of rational functions over a Hilbertian field $K$, i.e., $K(x)$ is Hilbertian. This property implies that we can apply the Hilbert's irreducibility theorem for the localized field $\mathbb{Q}(U)[A,X\setminus U]$ by a subset $U$ of $X$.  
\end{Remark}

To consider the density of such $\mathcal{O}$ in Euclidean topology, the concept of Hilbert subset is introduced as follows.

\begin{Definition}[Hilbert Subset]
    Let $f$ be an irreducible polynomial of degree $\ge 1$ in $\mathbb{Q}[A,X]$. Then, the set of rational points  $\alpha$ such that $f(\alpha,X)$ is irreducible over $\mathbb{Q}$ is called a basic Hilbert subset. The finite intersection of basic Hilbert subsets and Zariski-open sets is called a Hilbert subset. 
\end{Definition}

\begin{Remark}
    The finite intersection of Hilbert subsets is also a Hilbert subset by the definition. 
\end{Remark}

We may say that almost points in $\mathbb{Q}^m$ are in $\mathcal{O}$ by Lemma \ref{lem:dense}. 
   
\begin{Lemma}[\cite{Lang}, Corollary 2.5]\label{lem:dense}
     A Hilbert subset is dense for the ordinary topology. 
\end{Lemma}

We can utilize Lemma \ref{lem:fact} to check the primarity of a $0$-dimensional prime ideal. 

\begin{Lemma} \label{lem:mini-prime}
    Let $I$ be a zero-dimensional ideal in $K[X]$. Then, $I$ is a prime ideal if and only if its minimal polynomial in generic position is irreducible over $\mathbb{Q}$. Moreover, $I$ is a primary ideal if and only if its minimal polynomial in generic position is a power of an irreducible polynomial over $\mathbb{Q}$. 
\end{Lemma}

The next lemma refers to the relationship between the primarity of higher dimensional ideals and zero-dimensional ideals. 

\begin{Lemma}[\cite{Noro2004} Lemma 3.1, \cite{greuel2002singular} Chapter 4.3] 
\label{lem:cont-prime}
    Let $I$ be an ideal of $K(U)[X\setminus U]$ with $U$ is a maximal independent set of $X$ and $I^c=I\cap K[X]$, the contraction of $I$. Then:
    \begin{enumerate}
        \item If $I$ is a radical ideal, then $I^c$ is also a radical ideal.
        \item If $I$ is a prime ideal, then $I^c$ is also a prime ideal.
        \item If $I$ is a primary ideal, then $I^c$ is also a primary ideal.
    \end{enumerate}
Conversely, if $I$ is a prime (primary) ideal of $K[X]$ and $U$ is its MIS, 
then the extension $I^e$, which is the ideal generated by $I$ in $K(U)[X\setminus U]$, is 
also a prime(primary) ideal. Moreover, in this case, $(I^e)^c=I$ holds. 
\end{Lemma}

We generalize the Hilbert's irreducibility theorem for a prime ideal as follows. 

\begin{Theorem} \label{thm:hilb_prime}
    Let $P$ be a prime ideal $\mathbb{Q}[A,X]$. Then, there exist infinitely many points $\alpha\in \mathbb{Q}^m$ such that the specialized polynomial $\varphi_\alpha(P)$ is a prime ideal in $\mathbb{Q}[X]$. Moreover, there exist a dense set $\mathcal{O}$ of $\mathbb{Q}^m$ and a non-empty locally closed set $C$ such that $\varphi_\alpha(P)$ is a prime ideal in $\mathbb{Q}[X]$ for any $\alpha\in \mathcal{O}\cap C$. 
\end{Theorem}

\begin{proof}
    First, we assume $P$ is a local $0$-dimensional ideal. Then, by Lemma \ref{lem:mini-prime}, the local minimal polynomial $f$ in generic position with respect to $P$ is irreducible over $\mathbb{Q}$. By Theorem \ref{Thm:HIT}, there exists a Zariski dense set $\mathcal{O}$ of $\mathbb{Q}^m$ such that $f(\alpha,X)$ is irreducible in $\mathbb{Q}[X]$ if and only if $\alpha\in \mathcal{O}$. Here, there exists a non-empty locally closed set $C$ such that $f(\alpha,X)$ is a minimal polynomial in generic position with respect to $\varphi_\alpha (P)$, i.e., $\varphi_\alpha (P)$ is a prime ideal for any $\alpha\in C$. Therefore, $\varphi_\alpha(P)$ is a prime ideal in $\mathbb{Q}[X]$ for any $\alpha\in \mathcal{O}\cap C$. 

    Second, we consider $P$ of local dimension $\ge 1$. Let $U\subset X$ be a local maximal independent set of $P$. Then, there exists a non-empty locally closed set $C_1$ such that $U$ is a maximal independent set of $\varphi_\alpha (P)$ for any $\alpha \in C_1$. Then, for such $\alpha$, $K(U)[X\setminus U]\varphi_\alpha (P)$ is a zero-dimensional ideal in $P^e$ and a prime ideal for any $\alpha \in \mathcal{O}\cap C_2$, where $\mathcal{O}$ is a Zariski dense set of $\mathbb{Q}^m$ and $C_2$ a non-empty locally closed set. Also, there exists a non-empty locally closed set $C_3$ such that $\varphi_\alpha(P)=K(U)[X\setminus U]\varphi_\alpha (P)\cap K[X]=\varphi_\alpha(K(U)[X\setminus U]P\cap K[X])$ for any $\alpha \in C_3$ by $(O8)$ in Proposition \ref{prop:ideal_op}.
    
    Thus, by Lemma \ref{lem:cont-prime} and letting $C=C_1\cap C_2\cap C_3$, we obtain that $\varphi_\alpha(P)$ is a prime ideal in $\mathbb{Q}[X]$ for any $\alpha\in \mathcal{O}\cap C$. 
\end{proof}

A similar proof method yields the following corollary.

\begin{Corollary} \label{cor:hilb_primary}
    Let $P$ be a primary ideal $\mathbb{Q}[A,X]$. Then, there exist infinitely many points $\alpha\in \mathbb{Q}^m$ such that the specialized polynomial $\varphi_\alpha(Q)$ is a primary ideal in $\mathbb{Q}[X]$. Moreover, there exist a dense set $\mathcal{O}$ of $\mathbb{Q}^m$ and a non-empty locally closed set $C$ such that $\varphi_\alpha(Q)$ is a primary ideal in $\mathbb{Q}[X]$ for any $\alpha\in \mathcal{O}\cap C$. 
\end{Corollary}

Finally, we formalize primary decomposition system with Hilbert subset which is an extended notion of our feasible CPDS.  

\begin{Theorem}[Density of feasible CPDS] \label{thm:denseCPDS}
    Let $I$ be a proper ideal of $\mathbb{Q}[A,X]$. There exists a system $\{(C_1,\mathcal{O}_1,\mathcal{Q}_1),\ldots,(C_l,\mathcal{O}_l,\mathcal{Q}_l)\}$ of tuples of a constructible set, a Hilbert subset, and the reduced Gr\"obner basis such that $\varphi_\alpha (\mathcal{Q}_i)$ is a minimal primary decomposition of $\varphi_\alpha (I)$ for any $\alpha \in C_i\cap \mathcal{O}_i$. In particular, if $C_i$ is Zariski open then $C_i\cap \mathcal{O}_i$ is dense in the Euclidean topology. 
\end{Theorem}

\begin{proof}
    It is follows that there exists a feasible minimal CPDS $\{(C_1,\mathcal{Q}_1),$ $\ldots$ $,(C_l,\mathcal{Q}_l)\}$ of $I$. Then, for each $i$ and $Q_j\in \mathcal{Q}_i$, there exist a Hilbert set $\mathcal{O}_{ij}$ and a non-empty locally closed set $C_{ij}$ such that $\varphi_\alpha(Q_j)$ is a primary ideal in $\mathbb{Q}[X]$ for any $\alpha\in \mathcal{O}_{ij}\cap C_{ij}$ by Corollary \ref{cor:hilb_primary}. We may assume that $C_{ij}\subset C_{i}$. Then, letting $\mathcal{O}_{i}=\bigcap_j \mathcal{O}_{ij}$, $\{(C_1,\mathcal{O}_1,\mathcal{Q}_1),\ldots,(C_l,\mathcal{O}_l,\mathcal{Q}_l)\}$ satisfies the condition. 
\end{proof}

We state that for ``almost'' $\alpha \in \mathbb{Q}^m$, a minimal primary decomposition of $I$ with a genric condition ideal gives a minimal primary decomposition of $\varphi_\alpha (I)$ as follows.

\begin{Corollary}\label{cor:generic}
    Let $I$ be an ideal of $K[A,X]$ with a generic condition ideal and $\mathcal{Q}=\{Q_1,\ldots,Q_r\}$ a minimal primary decomposition of $I$. There exists a Hilbert subset $C$ such that $\varphi_\alpha (\mathcal{Q})=\{\varphi_\alpha (Q_1),\ldots,\varphi_\alpha (Q_r)\}$ is a minimal primary decomposition of $\varphi_\alpha (I)$ for any $\alpha \in C$. In particular, for almost every $\alpha\in \mathbb{Q}^m$, $\varphi_\alpha (\mathcal{Q})=\{\varphi_\alpha (Q_1),\ldots,\varphi_\alpha (Q_r)\}$ is a minimal primary decomposition of $\varphi_\alpha (I)$. 
\end{Corollary}

Finally, we give an explicit expression of $\mathcal{O}_i$ as follows.

\begin{Proposition}\label{pro:O-expression}
    Let $(C,\mathcal{O},\mathcal{Q})$ be an element of the system in Theorem \ref{thm:denseCPDS} and $\mathcal{Q}=\{Q_1,\ldots,Q_r\}$. One can compute local MISs $U_1,\ldots,U_r\subset X$ and local minimal polynomials $f_1(A,t)\in \mathbb{Q}(U_1)[A,t],\ldots,f_r(A,t)\in \mathbb{Q}(U_r)[A,t]$ corresponding to primary components 
    $Q_1,\ldots,Q_r$ such that
    \[
    \mathcal{O}=\{\alpha\in \mathbb{Q}^m\mid f_k(\alpha,t) \text{ is irreducible in } 
    K(U_k)[A,t] \text{ for all }k=1,\ldots,r\}.
    \]
\end{Proposition}

\begin{proof}
    For each $Q_i$, by Corollary \ref{cor:hilb_primary}, there exist a local MIS $U_i\subset X$ and a local minimal polynomial $f_i(A,t)\in \mathbb{Q}(U_i)[A,t]$ such that $\varphi_\alpha(Q_i)$ is a primary ideal for $\alpha\in \{\beta\in \mathbb{Q}^m\mid f_i(\beta,t) \text{ is irreducible in $\mathbb{Q}(U_k)[A,t]$}\}$. Thus, 
    for $\alpha\in \mathcal{O}$, $\varphi_\alpha (Q_1),\ldots,\varphi_\alpha(Q_r)$ are all primary ideals.
\end{proof}

\begin{algorithm}
\label{alg3}  
\begin{algorithmic}[1]      
\REQUIRE $(C,\mathcal{Q})$: a tuple of a constructible set and a set of primary ideals of $\mathbb{Q}[A,X]$ with order $r$
\ENSURE polynomials $f_1,\ldots,f_r$ such that $\varphi_\alpha(Q_1),\ldots,\varphi_\alpha(Q_r)$ are all primary ideals for any $\alpha \in \mathcal{Q}\cap C$, where $\mathcal{Q}=\{\beta\in C\mid f_k(\beta,t) \text{ is irreducible in } 
    \mathbb{Q}(U_k)[A,t] \text{ for all }k=1,\ldots,r\}$ 
\FOR{$i=1$ to $r$}
    \STATE $U_i\gets $ a local MIS  of $Q_i$ on $C$
    \STATE $f_i(A,t)\in \mathbb{Q}(U_i)[A,t]\gets$ a local minimal polynomial  of $Q_i$ on $C$
\ENDFOR
\RETURN $f_1,\ldots,f_r$
\caption{Defining Polynomials for a Hilbert Subset (\textsc{HilbertSubset})}
\end{algorithmic}
\end{algorithm}

By combining Algorithm~\ref{alg2} and Algorithm~3, we obtain Algorithm~\ref{alg4} as follows. 
\begin{algorithm}[H]            
\caption{\textsc{HilbertCPDS}}         
\label{alg4}                          
\begin{algorithmic}[1]              
\REQUIRE $I$: an ideal of $\mathbb{Q}[A,X]$
\ENSURE a minimal feasible CPDS of $I$ with Hilbert subsets
\STATE HilbertCPDS $\gets \{\}$
\STATE CPDS $\gets \textsc{MinimalFeasibleCPDS}(I)$
\FOR{$(C,\mathcal{Q})\in \text{CPDS}$}
    \STATE $F\gets \textsc{HilbertSubset}(C,\mathcal{Q})$
    \STATE HilbertCPDS $\gets \text{HilbertCPDS}\cup \{(C,F,\mathcal{Q})\}$
\ENDFOR
\RETURN HilbertCPDS
\end{algorithmic}
\end{algorithm}

\section{Computer Experiments}

We implement our algorithm in the computer algebra system Risa/Asir \cite{risaasir}. We use a library ``noro\_pd'' to compute ideal operations (intersection, saturation), primary decompositions, radicals of ideals. Also, we use a function kcgs\_main, one of the efficient programs, implemented by Katsusuke Nabeshima \footnote{\url{https://www.rs.tus.ac.jp/~nabeshima/softwares.html}} to compute comprehensive Gr\"obner systems of ideals. Timings in seconds are measured on a PC with AMD Ryzen Threadripper PRO 5965WX 24-Cores and 128GB memory. 

\subsection{Timings of Several Examples}

\begin{Example}
    We let $I_1=\langle x^2-a,bxy\rangle$ be an ideal of $\mathbb{Q}[a,b,x,y]$. In Example \ref{ex:alg}, it has a minimal  feasible CPDS $\{(\mathbb{Q}^2\setminus V(\langle ab\rangle),\{\langle x^2-a,y\rangle\}),$ $(V(\langle a\rangle)\setminus V(\langle a,b\rangle),\{\langle x\rangle,\langle x^2,y\rangle\})$ $,(V(\langle a,b\rangle),\{\langle x^2\rangle\})$,$(V(\langle b\rangle),\{\langle x^2-a\rangle\})\}$.
\end{Example}

\begin{Example}
 We consider  ``cyclic $n$-roots ideal'' ($cyclic(n)$) which is defined in \cite{Backelin}. For $cyclic(4)=\langle c_3c_2c_1c_0-1,((c_2+c_3)c_1+c_3c_2)c_0+c_3c_2c_1,(c_1+c_3)c_0+c_2c_1+c_3c_2,c_0+c_1+c_2+c_3\rangle$, if we take $c_0$ as a parameter, then we obtain its feasible CPDS consisting of three elements as follows.
\begin{align*}
&\{\\
(&\mathbb{Q}\setminus V(c_0^5-c_0),\langle c_0c_3-1,c_2+c_0,c_1+c_3\rangle,\langle c_0c_3+1,c_2+c_0,c_1+c_3\rangle),\\
(&V(c_0^2+1),\{\langle c_0^2+1,c_3-c_0,c_2^2+2c_0c_2-1,c_1+c_2-2c_0\rangle,\\
&\langle c_0^2+1,c_1-c_0,c_3^2+2c_0c_3-1,c_2+c_3+2c_0\rangle\}),\\
(&V(c_0+1),\{\langle c_0+1,c_3+1,c_2^2-2c_2+1,c_1+c_2-2\rangle,\langle c_0+1,c_3^2-2c_3+1,c_2+c_3-2,c_1+1\rangle\}),\\
(&V(c_0-1),\{\langle c_0-1,c_3-1,c_2^2+2c_2+1,c_1+c_2+2\rangle,\langle c_0-1,c_3^2+2c_3+1,c_2+c_3+2,c_1-1\rangle\}),\\
(&V(c_0),\{\langle 1\rangle\})\\
\}
\end{align*}
\end{Example}

\begin{Example}
     We define an adjacent-minor ideal of type $(2,3,4)$ $A_{(2,3,4)}=\langle x_{12}x_{21}-x_{11}x_{22},x_{13}x_{22}-x_{12}x_{23},x_{14}x_{23}-x_{13}x_{24},x_{22}x_{31}-x_{21}x_{32},x_{23}x_{32}-x_{22}x_{33},x_{24}x_{33}-x_{23}x_{34}\rangle\subset \mathbb{Q}[x_{ij}\mid 1\le i\le 3,1\le j\le 4]$, where a primary decomposition of an adjacent-minor ideal has important meanings in Computer Algebra for Statistics (see Section 5.3 of Chapter 5 in \cite{strumfels}). We compute a feasible CPDS of $A_{(2,3,4)}$ for each set of parameters. 
\end{Example}

	\begin{table}[H]
		\begin{center}
\begin{tabular}{|ll|c|}\hline
ideal & parameters & {\tt para\_pd}\\\hline
$I_1$ &$\{a,b\}$   &  0.03\\ 
$cyclic(4)$& $\{c_0\}$ &0.11 \\ 
 $A_{2,3,4}$& $\{x_{11}\}$&0.84 \\ 
  $A_{2,3,4}$& $\{x_{11},x_{12}\}$&0.75 \\ 
   $A_{2,3,4}$& $\{x_{11},x_{12},x_{13}\}$&1.08 \\ 
    $A_{2,3,4}$& $\{x_{11},x_{12},x_{13},x_{14}\}$&2.35 \\ 
     $A_{2,3,4}$& $\{x_{11},x_{12},x_{13},x_{14},x_{21}\}$&10.3 \\ 
\hline
\end{tabular}
			\caption{Timings of minimal feasible CPDSs} \label{table1}
		\end{center} 
	\end{table}	

In Table \ref{table1}, we see that our algorithm compute feasible CPDS effectively. Also, for $A_{2,3,4}$, the computation time increases as the number of parameters increases. The reason for this is that the computational complexity of CGS is highly dependent on the number of parameters.

\subsection{Simple Example for Hibert Subset of Prime Parametric Ideal}
In general, each Hilbert subset $\mathcal{O}_i$ given in Proposition \ref{pro:O-expression} 
can not be expressed explicitly as a locally closed set. 
But, for some local minimal polynomial $f_i(A,t)$, 
we can apply the naive method proposed in \cite{Yokoyama2006,Yokoyama2022} 
to $f_i(A,X)$ to obtain the expression of $\mathcal{O}_i$ as a locally closed set. 
\begin{Example} \label{ex:hilb}
Here we give a simple example of a prime parametric ideal for which 
an Hilbert subset can be computed in algebraic way. 
Consider the following ideal $P$ of $\mathbb{Q}[A,X]$, where $X=\{x_1,x_2,x_3\}$ and $A=\{a_1,a_2\}$.
\begin{eqnarray*}
P&= & \langle -x_1+x_2-a_1 x_3^2+a_2 x_3,\\
& & (a_1+1)x_1^3+(-x_2+(a_1^2-a_2)x_3+(a_2+1)a_1+1)x_1^2\\
& & +((-a_1-2)x_2+(a_1^2-a_2 a_1-2a_2)x_3 +a_2 a_1+a_2)x_1 \\
& & +x_2^2+((-a_1^2+2a_2)x_3-a_2 a_1-a_2)x_2+a_2^2 x_3^2+(-a_2^2 a_1-a_2^2)x_3,\\
& & -x_1^3+(x_3^2-a_1x_3-a_2-1)x_1^2+(x_2+x_3^2+(-a_1+a_2)x_3-a_2)\\
& & x_1+(-x_3^2+a_1 x_3+a_2)x_2
-a_2 x_3^3+a_2 x_3^2+a_2^2 x_3 \rangle 
\end{eqnarray*}
Then 
its Gr\"obner basis $G$ 
with respect to a block ordering $\{x_2\succ x_1\succ x_3\}\succ \succ \{a_2\succ a_1\}$
is 
\begin{eqnarray*}
&&\{g_1=x_1^3+(-x_3^2+a_1 x_3+a_2)x_1^2-a_1 x_3^2 x_1+a_1 x_3^4-a_1^2 x_3^3-a_2 x_3^2,\\
&& \;\; g_2=x_2-x_1-a_1 x_3^2+a_2 x_3\}. 
\end{eqnarray*}
By the form of $G$, it can be examined that $P$ is prime and $U=\{a_1,a_2,x_3\}$ is its MIS. 
Since the leading coefficients of elements in $G$ are 1 (constant) as polynomials over $\mathbb{Q}(A)$, 
$G$ is stable for any value $\alpha=(\alpha_1,\alpha_2)$ 
in $\mathbb{Q}^2$ and $U$  is also stable as an MIS.  
Moreover, for any value $\alpha=(\alpha_1,\alpha_2)$ in $\mathbb{Q}^2$, 
$\{x_3\}$ is an MIS for 
$\varphi_\alpha(P)=\langle \varphi_\alpha(g_1),\varphi_\alpha(g_2)\rangle$ 
and the minimal polynomial of $x_1$ modulo the extension $\varphi_\alpha(P)^e$ 
in $\mathbb{Q}(x_3)[x_1,x_2]$ is $\varphi_\alpha(g_1)=g_1(\alpha_1,\alpha_2,x_1,x_3)$. 
Then, the degree of $\varphi_\alpha(g_1)$ coincides with the linear dimension of the residue class ring $\frac{\mathbb{Q}(x_3)[x_1,x_2]}{\varphi_\alpha(P)^e}$ 
as a vector space over the rational function field $\mathbb{Q}(x_3)$, 
and $x_1$ is {\em in local generic position} with respect to $P$ for any value in $\mathbb{Q}^2$. 

Thus, for each value $\alpha$ in $\mathbb{Q}^2$, 
$\varphi_\alpha(P)$ is a prime ideal of $\mathbb{Q}[X]$ 
if and only if $\varphi_\alpha(g_1)$ is irreducible over $\mathbb{Q}$. 
For $g_1$, the {\em naive method for parametric factorization} 
proposed in \cite{Yokoyama2006,Yokoyama2022} can be applied 
to show that $\varphi_\alpha(g_1)$ is reducible over $\mathbb{Q}$ 
only if $\alpha$ is a zero of 
$\langle a_1^2a_2-a_1 a_2,a_1a_2^3+a_1 a_2^2-a_2^3-a_2^2\rangle$, 
that is, $\alpha_1=1$ or $\alpha_2=0$ or $\alpha=(0,-1)$. 
(See Example 2.6 in \cite{Yokoyama2022}.) 
Thus, we have the following dense subset $\mathcal{O}$ 
of $\mathbb{Q}^2$ such that $\varphi_\alpha(P)$ is prime for any $\alpha$ in $\mathcal{O}$:
\[
\mathcal{O}=\mathbb{Q}^2\setminus (\{(1,\alpha_2)\mid \alpha_2\in \mathbb{Q}\}\cup 
\{(\alpha_1,0)\mid \alpha_1\in \mathbb{Q}\}\cup \{(0,-1)\})
\]
We remark that for $\alpha$ in $\mathbb{Q}^2$ such that $\varphi_\alpha(P)$ is not prime, 
$\varphi_\alpha(P)$ has the following primary decomposition:
\begin{itemize}
\item Case where $\alpha_1=1$: $\varphi_\alpha(P)$ has three prime components 
for any $\alpha_2\in \mathbb{Q}$;
\begin{eqnarray*}
&& \langle x_2-x_3^2+(\alpha_2-1)x_3,x_1-x_3\rangle,
\langle x_2-x_3^2+(\alpha_2 +1)x_3,x_1+x_3\rangle, \\
&& \langle x_2-2x_3^2+(\alpha_2+1)x_3+\alpha_2,x_1-x_3^2+x_3+\alpha_2\rangle.
\end{eqnarray*}
\item Case where $\alpha_2=0$: The decomposition depends on the shape of $\alpha_1$ as follows:
\begin{itemize}
\item Case where $\alpha_1$ is not square: 
$\varphi_\alpha(P)$ has two prime components;
\[\langle x_2-(\alpha_1+1)x_3^2+\alpha_1x_3,x_1-x_3^2+\alpha_1 x_3\rangle, \langle x_1^2-\alpha_1 x_3^2,x_2-x_1-\alpha x_3^2\rangle.
\]
\item Case where $\alpha_1$ is square, that is, $\alpha_1=\beta^2$ for some $\beta$ in $\mathbb{Q}$: 
$\varphi_\alpha(P)$ has three prime components;
\begin{eqnarray*}
&& \langle x_1-\beta x_3,x_2-\beta^2 x_3^2-\beta x_3\rangle, 
\langle x_1+\beta x_3,x_2-\beta^2 x_3^2+\beta x_3\rangle, \\
&& \langle x_2-(\beta^2+1) x_1-\beta^4 x_3,x_1-x_3^2+\beta^2 x_3\rangle. 
\end{eqnarray*}
\end{itemize}

\item Case where $\alpha=(0,-1)$: $\varphi_\alpha(P)$ has two prime components;
\[
\langle x_2-x_3-1,x_1-1\rangle, 
\langle x_1^2-x_3^2 x_1-x_3^2,x_2-x_1-x_3\rangle.  
\]
\end{itemize}

\end{Example}

\section{Conclusion and Future Works}

In this paper, we proposed a notion of a feasible comprehensive primary decomposition system (CPDS) of a parametric ideal, as the first attempt, based on the Hilbert's irreducibility theorem. We also gave effective algorithms for computing feasible CPDS. In a naive computer experiment, we examined a certain effectiveness of our method. 

As future works, we will make theoretical and practical studies. 

\begin{itemize}
    \item (Theoretical) We will examine the density of a Hilbert subset in each locally closed set, by which the feasibility of CPDS becomes more rigidly defined. We will try to compute more precise CPDS by combining parametric factorization in \cite{Yokoyama2022}. 
    \item (Practical) We will improve the practicality to handle larger size problems derived from pure mathematics or engineering. By recursive structure of our algorithm, we might meet combinatorial explosion which causes impractical computation. In order to overcome this difficulty, it is required to introduce efficient techniques such as pruning unnecessary steps in \cite{Nabeshima}. Moreover, to examine the real practicality, we will apply our algorithm to solve systems of parametric equations e.g. geometrical problems in \cite{Chou}. 
\end{itemize}

\section*{Acknowledgements}
This work was supported by JSPS KAKENHI Grant Number JP22K13901.


\end{document}